\documentclass[a4paper,11pt]{article}
\usepackage[margin=3cm]{geometry}
\usepackage[pdfpagelabels, pagebackref, naturalnames]{hyperref}
\usepackage{amssymb, amsmath, amsthm}
\usepackage{graphicx}
\usepackage{paralist}
\usepackage{caption,subcaption}

\usepackage[utf8]{inputenc}
\usepackage{ae}
\usepackage{times}
\usepackage{color}

\usepackage{soul}

\hypersetup{bookmarksdepth=section}

\theoremstyle{plain}
\newtheorem{definition}{Definition}
\newtheorem{theorem}{Theorem}
\newtheorem{lemma}{Lemma}
\newtheorem{observation}{Observation}
\newtheorem{proposition}{Proposition}

\newtheorem{condition}{Condition}
{%\theorembodyfont{\rmfamily}

}
\newtheorem{claim}{Claim}

\newcommand{\ADD}[1]{#1}
\newcommand{\REMOVE}[1]{}

\DeclareMathOperator{\length}{length}
\DeclareMathOperator{\bend}{bend}

\renewcommand{\figurename}{Fig.}
\graphicspath{{pic/}}

\newenvironment{pf}{\begin{proof}}{\end{proof}}

\newcommand{\true}{\ensuremath\mathtt{true}}

\newcommand{\vect} {\mathbf}

\newcommand{\TT}{\mathrm{T}}
\newcommand{\BB}{\mathrm{B}}
\newcommand{\LL}{\mathrm{L}}
\newcommand{\RR}{\mathrm{R}}

\begin{document}

\title{Multi-Sided Boundary Labeling\thanks{A preliminary version of
    this paper has appeared in \emph{Proc.  13th Int. Algorithms Data
      Struct. Symp. (WADS'13)}, volume 8037 of \emph{Lect. Notes
      Comput. Sci.}, pages 463--474, Springer-Verlag.  This research
    was initiated during the GraDr Midterm meeting at the TU Berlin in
    October 2012.  The meeting was supported by the ESF EuroGIGA
    networking grant.  Ph.~Kindermann acknowledges support by the ESF
    EuroGIGA project GraDR (DFG grant Wo~758/5-1).}}

\author{%
Philipp~Kindermann\thanks{Lehrstuhl f\"ur Informatik I, Universit\"at
W\"urzburg, Germany.
    Homepage: http://www1.informatik.uni-wuerzburg.de/en/staff}
  \and
  Benjamin~Niedermann\thanks{Fakult{\"a}t f{\"u}r Informatik,
    Karlsruher Institut f\"ur Technologie (KIT), Germany.
    Email: \{rutter,$\,$niedermann\}@kit.edu}
  \and
  Ignaz~Rutter\footnotemark[3]
  \and
  Marcus~Schaefer\thanks{College of Computing and Digital Media,
    DePaul University, Chicago, IL, USA.
    Email: \mbox{mschaefer@cs.depaul.edu}}
  \and
Andr{\'e} Schulz\thanks{Institut f\"ur Mathematische Logik und
Grundlagenforschung, Universit\"at M\"unster, Germany.
  Email: \mbox{andre.schulz@uni-muenster.de}}
  \and
  Alexander~Wolff\footnotemark[2]
}

\date{\begin{minipage}{.5\textwidth}\raggedleft Submitted to Algorithmica in May 2014,\\
  revised in April 2015,\\
  accepted in July 2015.\end{minipage}}

% \keywords{%
%   Computational geometry, boundary labeling, dynamic program.}

\maketitle

\begin{abstract}
  In the \emph{Boundary Labeling} problem, we are given a set of~$n$
  points, referred to as \emph{sites}, inside an axis-parallel
  rectangle~$R$, and a set of~$n$ pairwise disjoint rectangular labels
  that are attached to~$R$ from the outside. The task is to connect
  the sites to the labels by non-intersecting rectilinear paths,
  so-called \emph{leaders}, with at most one bend.

  In this paper, we study the \emph{Multi-Sided Boundary Labeling}
  problem, with labels lying on at least two sides of the enclosing
  rectangle.  We present a polynomial-time algorithm that computes a
  crossing-free leader layout if one exists.  So far, such an
  algorithm has only been known for the cases in which labels lie on one
  side or on two opposite sides of~$R$ (here a crossing-free solution
  always exists).  The case where labels may lie on adjacent sides is
  more difficult.  We present efficient algorithms for testing the
  existence of a crossing-free leader layout that labels all sites and
  also for maximizing the number of labeled sites in a crossing-free
  leader layout.  For two-sided boundary labeling with adjacent sides,
  we further show how to minimize the total leader length in a
  crossing-free layout.
\end{abstract}

\section{Introduction}\label{sec:intro}

Label placement is an important problem in cartography and, more
generally, information visualization.  Features such as points, lines,
and regions in maps, diagrams, and technical drawings often have to be
labeled so that users understand better what they see.  Even very
restricted versions of the label-placement problem are
NP-hard~\cite{ksw-plsl-99}, which explains why labeling a map manually
is a tedious task that has been estimated to take 50\% of total map
production time~\cite{m-ctcc-80}.  The ACM Computational Geometry Impact
Task Force report~\cite{c-cgitf-99} identified label placement as an
important research area.  The point-labeling problem in particular has
received considerable attention, from practitioners and theoreticians
alike.  The latter have proposed approximation algorithms for various
objectives (label number versus label size), label shapes (such as
axis-parallel rectangles or disks), and label-placement models
(so-called fixed-position models versus slider models).

The traditional label-placement models for point labeling require that
a label is placed such that a point on its boundary coincides with the
point to be labeled, the \emph{site}.  This can make it impossible to
label all sites with labels of sufficient size if some sites are very
close together.  For this reason, Freeman et al.~\cite{fmc-alssm-96}
and Zoraster \cite{z-prusa-97} advocated the use of \emph{leaders},
(usually short) line segments that connect sites to labels.  In order
to ensure that the background image or map remains visible even in
the presence of large labels, Bekos et al.~\cite{bksw-blmea-07} took
a more radical approach.  They introduced models and algorithms for
\emph{boundary labeling}, where all labels are placed beyond the
boundary of the map and are connected to the sites by straight-line or
rectilinear leaders (see Fig.~\ref{fig:example}).

\newcommand{\mapScale}{0.38}
\begin{figure}[tb]
\centering
  \begin{subfigure}[t]{.42\textwidth}
    \includegraphics[scale=\mapScale,page=1]{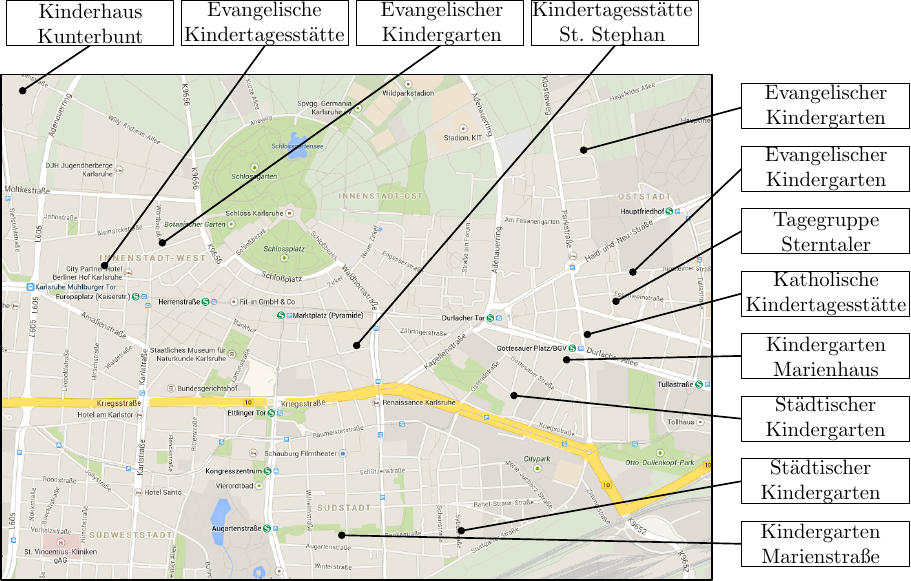}
    \caption{straight-line leaders}
    \label{sfg:original}
  \end{subfigure}
  \hfill
  \begin{subfigure}[t]{.425\textwidth}
    \hspace*{.5mm}%
    \raisebox{0.5mm}{\includegraphics[scale=\mapScale,page=2]{karte}}
    \caption{$opo$-leaders}
    \label{sfg:opo}
  \end{subfigure}
  \\
  \centering
  \begin{subfigure}[t]{.32\textwidth}
    \raisebox{0.5mm}{\includegraphics[scale=\mapScale,page=3]{karte}}
    \caption{$po$-leaders}
    \label{sfg:po}
  \end{subfigure}

  \caption{Labeling of kin\-der\-gar\-tens in Karlsruhe, Germany. The
    pictures show different types of leaders with labels on adjacent
    sides of the map.  For better readability, we have simplified the
    label texts.}
  \label{fig:example}
\end{figure}

\paragraph{Problem statement.}
\label{sub:definition}
Following Bekos et al. \cite{bksw-blmea-07}
% \todo{Marcus: I only had a quick look at the literature on boundary
%   labeling (btw, changed Bekos et al. [4, Fig 12] to Figure 16, I
%   don’t think 12 was the right picture, but I only had access to a pdf
%   version of the paper), but it seems to me that the boundary labeling
%   typically comes with an assignment between sites and labels. We
%   don’t really identify this as a separate variant in the
%   previous/related work section, and we don’t really have a strong
%   reference for the model we introduce in which no assignment is
%   given, do we? I’m just exposing my lack of knowledge of the
%   literature, I’m sure.}, 
we define the
\textsc{Boundary Labeling} problem as follows.  We are given an
axis-parallel rectangle $R=[0,W]\times[0,H]$, which is called the
\emph{enclosing rectangle}, a set~$P\subset R$ of~$n$ points
$p_1,\dots,p_n$, called \emph{sites}, within the rectangle~$R$,
and a set~$L$ of~$m \le n$ axis-parallel rectangles
$\ell_1,\dots,\ell_m$ of equal size, called
\emph{labels}.  The labels 
lie in the complement of~$R$ and touch the boundary of~$R$.  No two
labels overlap.  We denote an instance of the problem by the
triplet~$(R,L,P)$.  A \emph{solution} of a problem instance is a set
of~$m$ curves~$c_1,\dots,c_m$ in the interior of~$R$, called
\emph{leaders}, that connect sites to labels such that the leaders
% I assume we allow in general more labels than points?
\begin{inparaenum}[a)]
\item induce a matching between the labels and (a subset of) the sites,
\item touch the associated labels on the boundary of~$R$.
\end{inparaenum} %
Following previous work, we do not define labels as the text
  associated with the sites, but as the empty rectangles into which
  that text will be placed (during a post-processing step).  This
  approach is justified by our assumption that all label rectangles
  have the same size.

A solution is \emph{planar} if the leaders do not intersect.  We call
an instance \emph{solvable} if a planar solution exists. Note
that we do not prescribe which site connects to which label. The
endpoint of a curve at a label is called a \emph{port}.  We
distinguish two % incarnations
versions of the \textsc{Boundary Labeling}
problem: either the position of the ports on the boundary of~$R$ is
fixed and part of the input, or the ports \emph{slide}, i.e., their
exact location is not prescribed.

We restrict our solutions to \emph{$po$-leaders},
that is, starting at a site, the first line segment of a leader is
parallel ($p$) to the side of~$R$ touching the label it leads to,
and the second line segment is orthogonal ($o$) to that side; see
Fig.~\ref{sfg:po}.  (Fig.~\ref{sfg:opo} shows a labeling with
so-called \emph{$opo$-leaders}, which were investigated by Bekos et
al.~\cite{bksw-blmea-07}).  Bekos et al.~\cite[Fig.~16]{bkps-afbl-10}
observed that not every instance admits a planar solution
with $po$-leaders in which all sites are labeled (even if $m = n$).

% \begin{figure}[tb]
%   \centering
%   \includegraphics{definition}
%   \caption{Planar solution for an instance of \textsc{Two-Sided Boundary
%       Labeling with Adjacent Sides}.}
%   \label{fig:def}
% \end{figure}

\paragraph{Previous and related work.}
%
%\todo{Marcus, Previous and related work: already mentioned the Lubiw
%  article, which may be more relevant than the homotopic redrawing
%  article? All these would be for the model in which sites are
%  assigned to labels, I think.  (Schnipsel aus anderer Email:  (Chan,
%  Hoffmann, Kiazyk, Lubiw, Minimum Length Embedding of Planar Graphs
%  at Fixed Vertex Locations) gesehen? Das scheint eher relevant zu
%  sein, und erwaehnt noch ein paar andere interessante Artikel by
%  Liebling und Bastert/Fekete))}, 
For~$po$-labeling,
Bekos et al.~\cite{bksw-blmea-07} gave a simple
quadratic-time algorithm for the one-sided case that, in a first pass,
produces a labeling of minimum total leader length by matching sites
and ports from bottom to top.  In a second pass, their algorithm
removes all intersections without increasing the total leader length.
This result was improved by Benkert et al.~\cite{bhkn-amcbl-09} who
gave an~$O(n \log n)$-time algorithm for the same objective function
and an~$O(n^3)$-time algorithm for a very general class of objective
functions, including, for example, bend minimization.  They extend
the latter result to the two-sided case (with labels on
opposite sides of~$R$), resulting in an~$O(n^8)$-time algorithm.  For
the special case of leader-length minimization, Bekos et
al.~\cite{bksw-blmea-07} gave a simple dynamic program running in
$O(n^2)$ time.  All these algorithms work both for fixed and sliding
ports.

Leaders that contain a diagonal part have been studied by Benkert et
al.~\cite{bhkn-amcbl-09} and by Bekos et al.~\cite{bkns-blol-10}.
Recently, N\"ollenburg et al.~\cite{nps-d1sbl-10} have investigated a
dynamic scenario for the one-sided case, Gemsa et
al.~\cite{ghn-blapi-11} have used multi-layer boundary labeling to
label panorama images, and Fink et al.~\cite{fhssw-alfr-InfoVis12}
have boundary labeled focus regions, for example, in interactive
on-line maps.  Lin et al.~\cite{lky-m1bl-08} consider boundary
labeling where more than one site may be labeled by the same label.
Lin~\cite{l-cfmob-10} and Bekos et al.~\cite{bkfhknrs-mtoblwb-GD13} study
hyperleaders that connect each label to a set of sites.

At its core, the boundary labeling problem asks for a non-intersecting
perfect (or maximum) matching on a bipartite graph.  Note that an
instance may have a planar solution, although all of its leader-length
minimal matchings have crossings.  In fact, the ratio between a
length-minimal solution and a length-minimal crossing-free
matching can be arbitrarily bad; see Fig.~\ref{fig:crossmatching}.
When connecting points and sites with straight-line segments, the
minimum Euclidean matching is necessarily crossing-free.
For this case an $O(n^{2+\varepsilon})$-time  $O(n^{1+\varepsilon})$-space
algorithm exists~\cite{AES99}. 
% The minimum-length solution for connecting fixed
% pairs of points using rectilinear paths in the
% presence of obstacles is NP-hard, but there is a 2-approximation~\cite{SV10}.
%\AW{I removed Speckmann/Verbeek.  TODO: Mention Chan et
%  al. \cite{chkl-mlepg-GD13}, and possibly Liebing and Bastert/Fekete.}

\begin{figure}
  \centering
  \includegraphics[page=1,width=.2\textwidth]{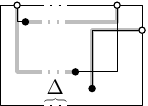}
  \caption{Length-mi\-ni\-mal solutions may have crossings. By
    increasing $\Delta$ we can make the ratio between the
    length-minimal matching and the length-minimal crossing-free
    matching arbitrarily small.}
  \label{fig:crossmatching}
\end{figure}

Boundary labeling can also be seen as a graph-drawing problem where
the class of graphs to be drawn is restricted to matchings.  The
restriction concerning the positions of the graph vertices (that is,
sites and ports) has been studied for less restricted graph classes
under the name \emph{point-set embeddability (PSE)}, usually following
the straight-line drawing convention for
edges~\cite{gmpp-eptvs-AMM91}. 
For polygonal edges, Bastert and Fekete~\cite{bf-gv-96}
  proved that PSE with minimum number of bends or minimum
  total edge length is NP-hard, even when the
  graph is a matching.  For minimizing the total edge length and the
  same graph class, Liebling
  et al.~\cite{lmmps-dpp} introduced heuristics and Chan
  et al.~\cite{chkl-mlepg-GD13} presented approximation algorithms.
  Chan et al.\ also considered paths and general planar graphs.  
PSE has also been combined with the
ortho-geodesic drawing convention~\cite{kkrw-mgepg-GD09}, which
generalizes~$po$-labeling by allowing edges to have more than one
bend.  The case where the mapping between ports and sites is given has
been studied in VLSI layout \cite{rcs-sbw-JA86}.

\paragraph{Our contribution.}
In the first part of the paper, we investigate the problem \textsc{Two-Sided
Boundary Labeling
with Adjacent Sides} where all labels lie on two \emph{adjacent}
sides of~$R$, without loss of generality, on the top and right side.
Note that point data often comes in a
coordinate system; then it is natural to have labels on adjacent
sides (for example, opposite the coordinate axes).  We argue that this
problem is more
difficult than the case where labels lie on opposite sides, which has
been studied before: with labels on opposite sides, (a)~there is
always a solution where all sites are labeled (if $m=n$) and (b)~a
feasible solution can be obtained by considering two instances of the
one-sided case.

We present an algorithm that, given an instance with~$n$ labels with
fixed ports and~$n$ sites, decides whether a 
planar solution exists where all sites are labeled and, if yes,
computes a layout of the leaders
(see Section~\ref{sec:two-sided-algorithm}).  Our algorithm uses
dynamic programming to ``guess'' a partition of the sites into the two
sets that are connected to the leaders on the top side and on the
right side.  The algorithm runs in~$O(n^2)$ time and uses~$O(n)$
space.

We study several extensions of our main result (see
Section~\ref{sec:two-sided-extensions}).  First, we show that our
approach for fixed ports can also be made to work for sliding ports.
Second, we \ADD{optimally} solve the label-number maximization problem (\ADD{in $O(n^3 \log n)$
time using $O(n)$ space}).  This is interesting if the position
of the sites and labels does not allow for a perfect matching or if
there are more sites than labels.
Finally, we present a modification of our algorithm that minimizes the
leader length (in~$O(n^8 \log n)$ time and~$O(n^6)$ space).

In the second part of the paper, we investigate the problems
\textsc{Three-Sided Boundary Labeling} and \textsc{Four-Sided Boundary
  Labeling} where the labels may lie on three or even all four sides
of~$R$, respectively.  To that end we generalize the concept of
partitioning the sites labeled by leaders of different sides. In this
way we obtain subinstances that we can solve using the algorithm for
the two-sided case.  We obtain an algorithm solving the
four-sided case in $O(n^9)$ time and~$O(n)$ space and an algorithm
 solving the three-sided case in $O(n^4)$ time and~$O(n)$ space.
Except for the leader-length minimization, all extensions presented
previously extend to the three- and four-sided case, of course with a
corresponding impact on the running~time \ADD{and space requirements}.

\paragraph{Notation.}
We call the labels that lie on the right (left/top/bottom) side of
$R$ \emph{right (left/top/bottom) labels}.  The \emph{type} of a
label refers to the side of~$R$ on which it is located.  The
\emph{type} of a leader (or a site) is simply the type of its label.
We assume that no two sites lie on the same horizontal or vertical
line, and no site lies on a horizontal or vertical line through a port
or an edge of a label.

For a solution~$\cal L$ of a boundary labeling problem, we define
several measures that will be used to compare different solutions.
We denote the total length of all leaders
in~$\cal L$ by~$\length(\mathcal{L})$.
Moreover, we denote
by~$|\mathcal L|_x$ the total length of all horizontal segments of
leaders that connect a left or right label to a site.  Similarly, we
denote by~$|\mathcal L|_y$ the total length of the vertical segments
of leaders that connect top or bottom labels to sites.  Note that
generally,~$|\mathcal L|_x + |\mathcal
L|_y \neq \length(\mathcal L)$.

We denote the (uniquely defined) leader connecting a site~$p$ to a
port~$t$ of a label~$\ell$ by~$\lambda(p,t)$.  We denote the bend of
the leader~$\lambda(p,t)$ by~$\bend(p,t)$.  In the case of fixed
ports, we identify ports with labels and simply write~$\lambda(p,\ell)$
and~$\bend(p,\ell)$, respectively.

\section{Structure of Two-Sided Planar Solutions}
\label{sec:two-sided-structure}

In this section, we \ADD{tackle} the two-sided boundary labeling problem
with adjacent sides by presenting a series of structural results of
increasing strength.  We assume that the labels are located on the top
and right sides of~$R$.  For simplicity, we assume that we have fixed ports.  By
identifying the ports with their labels, we use~$L$ to denote the set
of ports of all labels.  For sliding ports, we can simply fix all
ports to the bottom-left corner of their corresponding labels.  First
we show that a planar two-sided solution admits a transformation
sustaining planarity such that the result of the transformation can be
split into two one-sided solutions by constructing an
$xy$-monotone, rectilinear curve from the top-right to the bottom-left
corner of~$R$; see Fig.~\ref{fig:xysep}.  Afterwards, we provide a
necessary and sufficient criterion to decide whether there exists a
planar solution for a given separation.  This will form the basis of
our dynamic programming algorithm, which we present in
Section~\ref{sec:two-sided-algorithm}.

\begin{lemma}\label{lem:boxlemma}
  Consider a solution~$\mathcal{L}$ for~$(R,L,P)$ and let~$P'\subseteq
  P$ be sites of the same type.  Let~$L'\subseteq L$ be the
  set of labels of the sites in~$P'$.  Let~$K \subseteq R$ be a
  rectangle that contains all bends of the leaders of~$P'$.  If the
  leaders of~$P \setminus P'$ do not intersect~$K$, then we can
  rematch~$P'$ and~$L'$ such that the resulting
  solution~$\mathcal{L}'$ has the following properties:
  \begin{inparaenum}[(i)]
    \item all intersections in~$K$ are removed,
    \item there are no new intersections of leaders outside of~$K$,
    \item $|\mathcal{L}'|_x = |\mathcal{L}|_x$, $|\mathcal{L}'|_y =
      |\mathcal{L}|_y$, and
    \item $\length(\mathcal L') \le \length(\mathcal L)$.
  \end{inparaenum}
\end{lemma}

\begin{pf}
  Without loss of generality, we assume that~$P'$ contains top sites;
  the other cases are symmetric.  We first prove that, no matter how
  we change the assignment between~$P'$ and~$L'$, new intersection
  points can arise only in~$K$.  This enables us to construct the
  required solution.
  \begin{claim}\label{clm:boxclaim}
    Let~$\ell, \ell' \in L'$ and~$p, p' \in P'$ such that~$\ell$
    labels~$p$ and~$\ell'$ labels~$p'$.  Changing the matching by
    rerouting~$p$ to~$\ell'$ and~$p'$ to~$\ell$ does not introduce
    new intersections outside of~$K$.
  \end{claim}
  Let~$K'\subseteq K$ be the rectangle spanned by~$\bend(p,\ell)$
  and~$\bend(p',\ell')$.  When rerouting, we
  replace~$\lambda(p,\ell) \cup \lambda(p',\ell')$ restricted to
  the boundary of~$K'$ by its complement with respect to the
  boundary of~$K'$; see Fig.~\ref{fig:boxclaim} for an example.
  Thus, any changes concerning the leaders occur only in~$K'$. The
  statement of the claim follows.

\begin{figure}[tb]
  \centering
  %\begin{subfigure}[t]{.54\textwidth}
    \centering
    \includegraphics[page=2]{boxlemma}\hspace{1.0cm}%
    \includegraphics[page=1]{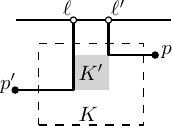}
   \hfill
%   \begin{subfigure}[t]{.42\textwidth}
%     \centering
%     \includegraphics[page=3]{boxlemma}\hfil%
%     \includegraphics[page=4]{boxlemma}
%     \caption{removing the highest crossing~$c$ does not increase the
%       total leader length}
%     \label{fig:remove-crossing}
%   \end{subfigure}
  \caption{Illustration of the proof of Lemma~\ref{lem:boxlemma}.
      Rerouting~$\lambda(p,\ell)$ and~$\lambda(p',\ell')$
      to~$\lambda(p,\ell')$ and $\lambda(p',\ell)$ changes leaders
      only on the boundary of~$K'$ }
  \label{fig:boxlemma}
      \label{fig:boxclaim}
\end{figure}

  Since any rerouting can be seen as a sequence of pairwise
  reroutings, the above claim shows that we can rematch~$L'$
  and~$P'$ arbitrarily without running the risk of creating new
  conflicts outside of~$K$.  To resolve the conflicts
  inside~$K$, we use the length-minimization algorithm for
  one-sided boundary labeling by Benkert et al.~\cite{bhkn-amcbl-09},
  with the sites and ports outside~$K$ projected onto the boundary
  of~$K$. Thus, we obtain a solution~$\mathcal L'$ satisfying
  properties~(i)--(iv).
\end{pf}

\begin{definition}\label{def:xysep}
  We call an $xy$-monotone, rectilinear curve connecting the top-right
  to the bottom-left corner of~$R$ an \emph{$xy$-separating curve}.
  We say that a planar solution to \textsc{Two-Sided Boundary Labeling
    with Adjacent Sides} is \emph{$xy$-separated} if and only if there
  exists an $xy$-separating curve~$C$ such that
  \begin{compactenum}[a)]
    \item the sites that are connected to the top side and all their
      leaders lie on or above~$C$
    \item the sites that are connected to the right side and all their
      leaders lie below~$C$.
  \end{compactenum}
\end{definition}

It is not hard to see that a planar solution is not $xy$-separated if
there exists a site~$p$ that is labeled to the right side and a
site~$q$ that is labeled to the top side with~$x(p)<x(q)$ and
$y(p)>y(q)$.  There are exactly four patterns in a possible planar
solution that satisfy this condition; see Fig.~\ref{fig:patterns}.  In
Lemma~\ref{lem:xysep}, we show that these patterns are the only ones
that can violate $xy$-separability.

\begin{figure}[tb]
  \begin{minipage}[b]{.3\textwidth}
    \centering
    \includegraphics{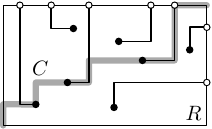}
  \end{minipage}
  \hfill
  \begin{minipage}[b]{.67\textwidth}
    \renewcommand{\thesubfigure}{P\arabic{subfigure}}
    \begin{subfigure}[b]{.23\textwidth}
      \centering
      \includegraphics[page=1,width=\linewidth]{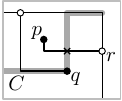}
      \caption{}
      \label{fig:pattern-1}
    \end{subfigure}
    \hfill
    \begin{subfigure}[b]{.23\textwidth}
      \centering
      \includegraphics[page=2,width=\linewidth]{patternscrossings}
      \caption{}
      \label{fig:pattern-2}
    \end{subfigure}
    \hfill
    \begin{subfigure}[b]{.23\textwidth}
      \centering
      \includegraphics[page=4,width=\linewidth]{patternscrossings}
      \caption{}
      \label{fig:pattern-3}
    \end{subfigure}
    \hfill
    \begin{subfigure}[b]{.23\textwidth}
      \centering
      \includegraphics[page=3,width=\linewidth]{patternscrossings}
      \caption{}
      \label{fig:pattern-4}
    \end{subfigure}
    \renewcommand{\thesubfigure}{\alpha{subfigure}}
  \end{minipage}

  \begin{minipage}[t]{.3\textwidth}
    \centering
    \caption{An $xy$-separating curve %$C$
      of a planar solution.}
    \label{fig:xysep}
  \end{minipage}\hfill
  \begin{minipage}[t]{.67\textwidth}
    \centering
    \caption{A planar solution that contains any of the above four
      patterns P1--P4 is not $xy$-separated.}
    \label{fig:patterns}
  \end{minipage}
\end{figure}

\newcommand{\lemXySepText}{A planar solution is $xy$-separated if and
  only if it does not contain any of the patterns P1--P4 in
  Fig.~\ref{fig:patterns}.}

\begin{lemma}\label{lem:xysep}
  \lemXySepText
\end{lemma}

\begin{pf}
  Obviously, the planar solution is not $xy$-separated if one of these
  patterns occurs.  Let us assume that none of these patterns exists.
  We construct an $xy$-monotone curve~$C$ from the top-right corner
  of~$R$ to its bottom-left corner.  We move to the left whenever
  possible, and down only when we reach the $x$-coordinate of a
  site~$p$ that is connected to the top, or when we reach the
  $x$-coordinate of a port of a top label, labeling a site~$p$.  If we
  have to move down, we move down as far as necessary to avoid the
  corresponding leader, namely down to the $y$-coordinate of~$p$.
  Finally, when we reach the left boundary of~$R$, we move down to the
  bottom-left corner of~$R$.  If $C$ is free of crossings, then we have
  found an $xy$-separating curve.  (For an example,
  see curve~$C$ in Fig.~\ref{fig:xysep}.)

  Assume for a contradiction, that a crossing arises during the
  construction, and consider the topmost such crossing.  Note that, by
  the construction of~$C$, crossings can only occur with leaders that
  connect a site~$p$ to a right port~$r$.  We distinguish two cases,
  based on whether the crossing occurs on a horizontal or a vertical
  segment of the curve~$C$.

  If~$C$ is crossed on a vertical segment, then this segment belongs to
  a leader connecting a site~$q$ to a top port~$t$, and we have
  reached the $x$-coordinate of either the port or the site.
  Had we, however, reached the $x$-coordinate of the port, this
  would imply a crossing between~$\lambda(p,r)$ and~$\lambda(q,t)$.
  Thus, we have reached the $x$-coordinate of~$q$.  This means
  that~$p$ lies to the left of and above~$q$, and we have found one
  of the patterns~\subref{fig:pattern-1} and~\subref{fig:pattern-2};
  see Fig.~\ref{fig:patterns}.

  If~$C$ is crossed on a horizontal segment, then~$p$ must lie
  above~$r$.  Otherwise, there would be another crossing of~$C$ with
  the same leader, which is above the current one.  This would
  contradict the choice of the topmost crossing.  Consider the
  previous segment of~$C$, which is responsible for reaching the
  $y$-coordinate of the current segment.  This vertical segment belongs
  to a leader connecting a site~$q$ to a top port~$t$.  Since leaders
  do not cross, $q$ is to the right of~$p$, and the crossing
  on~$C$ implies that~$q$ is below~$p$.  We have found one of the
  patterns~\subref{fig:pattern-3} and~\subref{fig:pattern-4}; see
  Fig.~\ref{fig:patterns}.
\end{pf}

Observe that patterns~\subref{fig:pattern-1}
and~\subref{fig:pattern-2} can be transformed into
patterns~\subref{fig:pattern-3} and~\subref{fig:pattern-4},
respectively, by mirroring the instance diagonally.  Next, we prove
constructively that, by rerouting pairs of leaders, any planar
solution can be transformed into an $xy$-separated planar solution.

\begin{proposition}\label{prop:solut}
  If there exists a planar solution~$\mathcal L$ to \textsc{Two-Sided
    Boundary Labeling with Adjacent Sides}, then there exists an
  $xy$-separated planar solution $\mathcal L'$ with~$\length(\mathcal
  L') \le \length(\mathcal L)$, $|\mathcal L'|_x \le |\mathcal L|_x$,
  and $|\mathcal L'|_y \le |\mathcal L|_y$.
\end{proposition}

\begin{pf}
  Let~$\mathcal L$ be a planar solution of minimum total leader
  length.  We show that $\cal L$ is $xy$-separated.  Assume, for the
  sake of contradiction, that ~$\cal L$ is not $xy$-separated.  Then,
  by Lemma~\ref{lem:xysep}, $\cal L$ contains one of the
  patterns~\subref{fig:pattern-1}--\subref{fig:pattern-4}.  Without
  loss of generality, we can assume that the pattern is of
  type~\subref{fig:pattern-3} or~\subref{fig:pattern-4}.  Otherwise,
  we mirror the instance diagonally.

  Consider all patterns~$(p,q)$ in $\mathcal L$ of
  type~\subref{fig:pattern-3} or~\subref{fig:pattern-4} such that $p$
  is a right site (with port~$r$) and~$q$ is a top site (with
  port~$t$).  Among all such patterns, consider the ones where~$p$ is
  rightmost and among these pick one where~$q$ is bottommost.
  Let~$A$ be the rectangle spanned by~$p$ and~$t$; see
  Fig.~\ref{fig:eliminate}..  Let~$A'$ be the rectangle spanned
  by~$\bend(q,t)$ and~$p$.  Let~$B$ be the rectangle spanned by~$q$
  and~$r$.  Let~$B'$ be the rectangle spanned by~$q$ and~$\bend(p,r)$.
  Then we claim the following:
  \begin{compactenum}[(i)]
  \item Sites in the interiors of~$A$ and~$A'$ are connected to the top.
  \item Sites in the interiors of~$B$ and~$B'$ are connected to the
    right.
  \end{compactenum}
  Property~(i) is due to the choice of~$p$ as the rightmost site
  involved in such a pattern.  Similarly, property~(ii) is due to the
  choice of~$q$ as the bottommost site that forms a pattern with~$p$.
  This settles our claim.

  Our goal is to change the labeling by rerouting~$p$ to~$t$ and~$q$
  to~$r$, which decreases the total leader length, but may introduce
  crossings.  We then use Lemma~\ref{lem:boxlemma} to remove the
  crossings without increasing the total leader length.
  Let~$\mathcal L''$ be the labeling obtained from~$\mathcal L$ by
  rerouting~$p$ to~$t$ and~$q$ to~$r$.  We have~$|\mathcal L''|_y \le
  |\mathcal L|_y - (y(p) - y(q))$ and~$|\mathcal L''|_x = |\mathcal
  L|_x - (x(q) - x(p))$.  Moreover,~$\length(\mathcal L'') \le
  \length(\mathcal L) - 2 (y(p) - y(q))$, as at least twice the
  vertical distance between~$p$ and~$q$ is saved; see
  Fig.~\ref{fig:eliminate}.  Since the original labeling was planar,
  crossings may only arise on the horizontal segment of~$\lambda(p,t)$
  and on the vertical segment of~$\lambda(q,r)$.

\begin{figure}
  \centering
    \begin{subfigure}{.3\linewidth}
      \centering
      \includegraphics[page=2]{eliminate}
      \caption{pattern~\subref{fig:pattern-3}}
    \end{subfigure}
    \begin{subfigure}{.3\linewidth}
      \centering
      \includegraphics[page=1]{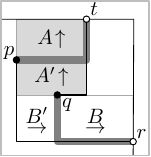}
      \caption{pattern~\subref{fig:pattern-4}}
    \end{subfigure}
    \caption{Types (top$\,$=$\,\uparrow$ / right$\,$=$\,\to$) of the
      sites inside rectangles~$A$, $A'$, $B$,
      and~$B'$.}
    \label{fig:eliminate}
\end{figure}

  By properties (i) and (ii), all leaders that cross the new
  leader~$\lambda(p,t)$ have their bends inside~$A'$, and all leaders
  that cross the new leader~$\lambda(q,r)$ have their bends inside~$B'$.
  Thus, we can apply Lemma~\ref{lem:boxlemma} to the rectangles~$A'$ and~$B'$
  to resolve all new crossings.  The resulting solution~$\mathcal L'$
  is planar and has length less than~$\length(\mathcal L)$.  This is a
  contradiction to the choice of~$\mathcal L$.
\end{pf}

Since every solvable instance of \textsc{Two-Sided Boundary Labeling
  with Adjacent Sides} admits an $xy$-separated planar solution,
it suffices to search for such a solution.  Moreover, an
$xy$-separated planar solution that
minimizes the total leader length is a solution of minimum length.  In
Lemma~\ref{lem:condit} we provide a necessary and sufficient criterion
to decide whether, for a given $xy$-monotone curve~$C$, there is a
planar solution that is separated by~$C$.  We denote the region of~$R$
above~$C$  by~$R_\TT$ and the region of~$R$ below~$C$
by~$R_\RR$. We do not include $C$ in either $R_\TT$ or~$R_\RR$, so these 
regions are open at~$C$.

For any point~\ADD{$a\in R$}, we define the
rectangle~\ADD{$R_a$}, spanned by the top-right corner of~$R$ and~\ADD{$a$}.
We define~\ADD{$R_a$} such that it is closed but does \emph{not} contain
its top-left corner.  In particular, we consider the port of a
top label as contained in~\ADD{$R_a$}, only if it is not the upper left
corner.

A rectangle~\ADD{$R_a$} is \emph{valid} if the number of sites of~$P$
above~$C$ that belong to~\ADD{$R_a$} is at least as large as the
number of ports on the top side of~\ADD{$R_a$}.  \ADD{The central idea is
that the labels on the top side of a valid rectangle~$R_a$ can be
connected to the sites in~$R_a$ by leaders that are completely
contained inside that rectangle.}
We are now ready to present the \emph{strip condition}.

\begin{condition}
  The \emph{horizontal strip condition} of the point~\ADD{$b\in C$} is satisfied if
  there exists a point~\ADD{$a \in R_\TT$} with~\ADD{$y(a)=y(b)$}
  and~\ADD{$x(a)\le x(b)$} such that~\ADD{$R_{a}$} is valid.
\end{condition}

Without loss of generality we may assume that the curve~$C$ is rectilinear.\REMOVE{ We
observe that for the horizontal strip condition, it is sufficient
to consider the horizontal segments of~$C$.} \ADD{The condition is named after the horizontal segments through points in~$C$.}

We now prove that, for a given $xy$-monotone curve~$C$ connecting the
top-right corner to the bottom-left corner of~$R$, there exists a
planar solution in~$R_\TT$ for the top labels if and only if
all points of~$C$ satisfy the strip condition.

\begin{lemma}
  \label{lem:condit}
  Let~$C$ be an $xy$-monotone curve from the top-right corner of~$R$
  to the bottom-left corner of~$R$.  Let~$P' \subseteq P$ be the sites
  that are in~$R_\TT$.  There is a planar solution that uses all top
  labels of~$R$ to label sites in~$P'$ in such a way that all leaders
  are in~$R_\TT$ if and only if each point of~$C$
  satisfies the strip condition.
\end{lemma}

\begin{pf}
For the proof we call a
region~$S\subseteq R$ \emph{balanced} if it contains the same number
of sites as it contains ports.
  To show that the conditions are necessary, let~$\mathcal L$ be a
  planar solution for which all top leaders are above~$C$.  Consider a
  point~\ADD{$b\in C$}.  If \ADD{$y(p)\ge y(b)$} for all sites~\ADD{$p\in P'$},
  rectangle~\ADD{$R_{a}$} with \ADD{$a=(0,y(b))$} is clearly valid, and thus the
  strip condition for~\ADD{$b$} is satisfied.  Hence, assume that there is a
  site~\ADD{$p \in P'$} with~\ADD{$y(p)<y(b)$} that is labeled by a top label;
  see Fig.~\ref{fig:strips:onesided:case1}.  Then, the vertical
  segment of this leader crosses the horizontal line~$h$
  through~\ADD{$b$}. Let~\ADD{$a$} denote the rightmost such crossing of a leader
  of a site in~$P'$ with~$h$.  We claim that~\ADD{$R_{a}$} is valid.  To see
  this, observe that all sites of~$P'$ top-right of~\ADD{$a$} are contained
  in~\ADD{$R_{a}$}.  Since no leader may cross the vertical segment
  defining~\ADD{$a$}, the number of sites in~\ADD{$R_{a}\cap R_\TT$} is balanced,
  i.e., \ADD{$R_{a}$} is valid.

  Conversely, we show that if the conditions are satisfied, then a
  corresponding planar solution exists.  For each horizontal segment
  of~$C$ consider the horizontal line through the segment. We denote
  the part of these lines within~$R$ by~$h_1,\dots,h_l$,
  respectively, and let~$h_0$ be the top side of~$R$. The
  line segments~$h_1,\dots,h_l$ partition~$R_\TT$ into~$l$
  strips, which we denote by~$S_1,\dots,S_l$ from top to bottom, such
  that strip~$S_i$ is bounded by~$h_i$ from below
  for~$i=1,\dots,l$; see Fig.~\ref{fig:strips:onesided:case2}.
  Additionally, we define~$S_0$ to be the empty strip that coincides
  with~$h_0$.  Let~$S_{k}$ be the last strip that contains sites of~$P'$.
  For~$i=0,\dots,k-1$, let~\ADD{$a_i'$} denote the rightmost point
  of~$h_i\cap R_\TT$ such that~\ADD{$R_{a_i'}$} is valid.
  Such a point exists since the leftmost point of~$h_i\cap C$
  satisfies the strip condition.
  We define~\ADD{$a_i$} to be the point on~$h_i \cap R_\TT$, whose
  $x$-coordinate is~\ADD{$\min_{j \le i}\{x(a_j')\}$}.  Note that~\ADD{$R_{a_i}$}
  is a valid rectangle, as, by definition, it completely contains some
  valid rectangle~\ADD{$R_{a_j'}$ with~$x(a_j') = x(a_i)$}.  Also by
  definition the sequence formed by the points~\ADD{$a_i$} has
  decreasing~$x$-coordinates, i.e., the \ADD{$R_{a_i}$} grow to the left;
  see Fig.~\ref{fig:strips:onesided:case3}.

\begin{figure}
    \begin{subfigure}{.32\linewidth}
      \centering
      \includegraphics[page=1]{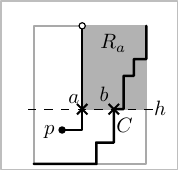}
      \caption{}\label{fig:strips:onesided:case1}
    \end{subfigure}
    \hfill
    \begin{subfigure}{.32\linewidth}
      \centering
      \includegraphics[page=2]{strips_one_sided}
      \caption{}\label{fig:strips:onesided:case2}
    \end{subfigure}
    \hfill
    \begin{subfigure}{.32\linewidth}
      \centering
      \includegraphics[page=3]{strips_one_sided}
      \caption{}\label{fig:strips:onesided:case3}
    \end{subfigure}
    \caption{The strip condition.
      \subref{fig:strips:onesided:case1})~The
      horizontal strip condition of~\ADD{$b$} is satisfied by~\ADD{$a$}.
      \subref{fig:strips:onesided:case2})~The
      horizontal segments of~$C$ partition the strips~$S_0,
      S_1,\ldots,S_k$.
      \subref{fig:strips:onesided:case3})~Constructing a planar
      labeling from a sequence of valid rectangles.}
\end{figure}

  We prove inductively that, for each~$i=0,\dots,k$, there is a planar
  labeling~$\mathcal L_i$ that matches the labels on the top side
  of~\ADD{$R_{a_i}$} to points contained in~\ADD{$R_{a_i}$}, in such a way that
  there exists an~$xy$-monotone curve~$C_i$ from the top-left
  to the bottom-right corner of~\ADD{$R_{a_i}$} that separates the
  labeled sites from the unlabeled sites without intersecting any
  leaders.  Then~$\mathcal L_k$ is the required labeling.

  For~$i=0$,~$\mathcal L_0 = \emptyset$ is a planar solution.
  Consider a strip~$S_i$ with~$0 < i \le k$; see
  Fig.~\ref{fig:strips:onesided:case3}.  By the induction hypothesis,
  we have a curve~$C_{i-1}$ and a planar labeling~$\mathcal L_{i-1}$,
  which matches the labels on the top side of~\ADD{$R_{a_{i-1}}$} to
  the sites in~\ADD{$R_{a_{i-1}}$} above~$C_{i-1}$.  To extend it to a
  planar solution~$\mathcal L_i$, we additionally need to match the
  remaining labels on the top side of~\ADD{$R_{a_i}$} and construct a
  corresponding curve~$C_i$.  Let~$P_i$ denote the set of unlabeled
  sites in~\ADD{$R_{a_i}$}.  By the validity of~\ADD{$R_{a_i}$}, this
  number is at least as large as the number of unused ports at the top
  side of~\ADD{$R_{a_i}$}. \ADD{We arbitrarily match these ports to
    the topmost sites of~$P_i$ that are not labeled in~$\mathcal
    L_{i-1}$. We denote the resulting labeling by $\mathcal L'_i$. We
    observe that no leader of~$\mathcal L'_i$ crosses the
    curve~$C_{i-1}$, and hence such leaders cannot cross leaders
    in~$\mathcal L_{i-1}$. Let~$h$ be the topmost horizontal line such
    that all labeled sites of~$\mathcal L'_i$ lie above $h$. Further,
    let~$K$ be the rectangle that is spanned by the top-left corner
    of~$R_{a_{i-1}}$ and the intersection of~$h$ with the left side
    of~$R_{a_i}$. Since the ports of~$\mathcal L'_i$ lie on the top
    side of~$K$, any leader's bend of $\mathcal L'_i$ lies in~$K$. We
    apply Lemma~\ref{lem:boxlemma} on $\mathcal L'_i$ to obtain a
    planar labeling $\mathcal L''_i$, which has no crossings with
    $\mathcal L_{i-1}$. Hence, the set $\mathcal L_i=\mathcal
    L''_{i}\cup \mathcal L_{i-1}$ is the required labeling.}

  It remains to construct the curve~$C_i$.  For this, we start at the
  top-left corner of~\ADD{$R_{a_i}$} and move down vertically, until
  we have passed all labeled sites.  We then move right until we
  either hit~$C_{i-1}$ or the right side of~$R$.  In the former case,
  we follow~$C_{i-1}$ until we arrive at the right side of~$R$.
  Finally, we move down until we arrive at the bottom-right corner
  of~\ADD{$R_{a_i}$}.  Note that all labeled sites are above~$C_i$,
  unlabeled sites are below~$C_i$, and no leader is crossed by~$C_i$.
  This is true since we first move below the new leaders and then
  follow the previous curve~$C_{i-1}$.
\end{pf}

A symmetric strip condition (with vertical strips) can be
obtained for the right region~$R_\RR$ of a partitioned instance.  The
characterization is completely symmetric.

In the following we observe two properties of the strip condition. The first observation
states that the horizontal strip condition
at~$(x,y)$ is independent of the exact shape of the curve between the top-right corner~$r$ of $R$
and~$(x,y)$, as
long as the number of sites above the curve remains the same.  This is crucial for
using dynamic programming to test the existence of a suitable curve.
The second observation states that the horizontal strip condition can
only be violated when the curve passes the $x$-coordinate of a top site.  This
enables us to discretize the problem.

\begin{observation}\label{obs:stripnumber}
  The horizontal strip condition for a point~\ADD{$a\in C$} depends only on the number
  of sites in~\ADD{$R_a$} above~$C$, in the following sense:
  Let~$C$ and~$C'$ be two $xy$-monotone curves from~$r$ to~\ADD{$a$}
  with~$u$ sites in~\ADD{$R_a$} above~$C$ and~$C'$, respectively.
  Then, \ADD{$a$} satisfies the strip condition for~$C$ if and only if it
  satisfies the strip condition for~$C'$.
\end{observation}

\begin{observation}\label{obs:stripwalk}
  Let~\ADD{$a,b\in C, x(a)\le x(b)$} such that there is no top site~$\ell$
  with~\ADD{$x(a)<x(\ell)\le x(b)$}. Then, \ADD{$a$} satisfies the horizontal strip
  condition for~$C$ if and only if~\ADD{$b$} satisfies the horizontal strip condition
  for~$C$.
\end{observation}

Symmetric statements hold for the vertical strip condition.  In the
following, we say that a point~$(x,y)$ on a curve~$C$ satisfies the
\emph{strip condition} if it satisfies both the horizontal and the
vertical strip condition.

\section{Algorithm for the Two-Sided Case}
\label{sec:two-sided-algorithm}

How can we find an $xy$-monotone curve~$C$ that satisfies
the strip conditions? For that purpose we only consider
$xy$-monotone curves contained in some graph $G$ that is dual to the
rectangular grid induced by the
sites and ports of the given instance.
Note that this is not a restriction since all leaders are contained in
the grid induced by the sites and ports. Thus, every $xy$-monotone curve
that does not intersect the leaders can be transformed into an
equivalent $xy$-monotone curve that lies on~$G$.

When traversing an edge~$e$ of~$G$, we pass the $x$- or $y$-coordinate
of exactly one entity of our instance; either a site (\emph{site event}) or
a port (\emph{port event}). When passing a site, the position of the
site relative to~$e$ (above/below $e$ or right/left of $e$)
decides whether the site
is connected to the top or to the right side. Clearly, there is an
exponential number of possible $xy$-monotone traversals through
the grid. In the following, we describe a dynamic program that
finds an $xy$-separating curve in~$O(n^3)$ time.

Let~$m_\RR$~and~$m_\TT$ be the numbers of ports on the
right and top side of~$R$, respectively. Also, let $N=n+m_\TT+2$ and
$M=n+m_\RR+2$, then the grid~$G$ has
size~$N\times M$. We define the grid
points as $G(s,t)$, $0\le s \le N$, $0\le t\le M$ with~$G(0,0)$ being the
bottom-left and~$r:=G(N,M)$ being the top-right
corner of~$R$. Finally, let~$G_x(s):=x(G(s,0))$
and~$G_y(t):=y(G(0,t))$.

\begin{figure}[tb]
  \centering
  \includegraphics{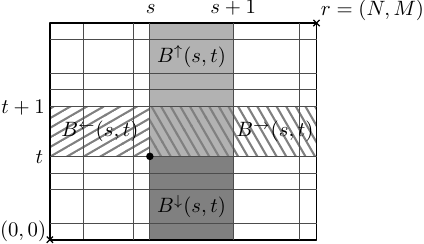}
  \caption{The four boxes~$B^\uparrow(s,t), B^\downarrow(s,t),
    B^\leftarrow(s,t)$ and~$B^\rightarrow(s,t)$ defined by
    grid point~$(s,t)$.}
  \label{fig:fourboxes}
\end{figure}

For each grid point~$(s,t)$ that is neither on the topmost row
nor on the rightmost column, we define four boxes~$B^\uparrow(s,t),
B^\downarrow(s,t), B^\leftarrow(s,t)$ and~$B^\rightarrow(s,t)$ as
follows; see Fig.~\ref{fig:fourboxes} for an illustration.

\begin{compactenum}
\item $B^\uparrow(s,t) = \{
(x,y) \in R \mid G_x(s) \le x \le G_x(s+1) \wedge y \ge G_y(t)\}$

\item $B^\downarrow(s,t) = \{ (x,y) \in R \mid G_x(s) \le x \le
  G_x(s+1) \wedge y \le G_y(t)\}$

\item $B^\leftarrow(s,t) = \{(x,y) \in R \mid G_y(t) \le y \le
  G_y(t+1) \wedge x \le G_x(s)\}$

\item $B^\rightarrow(s,t) = \{(x,y) \in R \mid G_y(t) \le y \le
  G_y(t+1) \wedge x \ge G_x(s)\}$
\end{compactenum}
We define a table~$T[(s,t),u,b]$ that assigns to each
grid position~$(s,t)$ and number of points~$u$ and~$b$ a Boolean value.
We define~$T[(s,t), u,b]$ to be $\true$ if and only if there
exists an $xy$-monotone curve~$C$ satisfying the following conditions.
\begin{compactenum}[(i)]
  \item Curve~$C$ starts at~$r$ and ends at~$G(s,t)$.
  \item Inside the rectangle spanned by~$r$ and~$G(s,t)$, there are~$u$
    sites of~$P$ above~$C$ and~$b$ sites of~$P$ below~$C$.
  \item For each grid point on~$C$, the strip condition holds.
\end{compactenum}

These conditions together with Proposition~\ref{prop:solut} and
Lemma~\ref{lem:condit} imply that the instance admits a planar solution if
and only if $T[(0,0),u,b]=\true$ for some $u$ and~$b$.

We define a Boolean function~$S[(s,t),u,b]$ that is true if and only
if the strip condition at~$(s,t)$ is satisfied for some $xy$-monotone
curve~$C$ (and thus by Observation~\ref{obs:stripnumber} for all such
curves) from~$r$ to~$G(s,t)$ with~$u$ sites above and~$b$ sites
below~$C$.
%Note that
%this information suffices by Observation~\ref{obs:stripnumber}.
%
The following lemma gives a recurrence for~$T$, which is essentially a
disjunction of two values, each of which is determined by distinguishing
three cases.

\begin{lemma}\label{lem:recurrence}
  For $s=N$ and~$t=M$, it holds that $T[(s,t),0,0]=\true$.
  For $s\in[0,N-1]$ and $t\in[0,M-1]$, it holds that

  \vspace{1em}

  \begin{tabular}{cl}
    & $\left\{\begin{array}{l@{\quad}l@{\quad}l}
    T[(s+1,t),u,b]\wedge S[(s,t),u,b] && L\cap B^\uparrow(s,t)\neq\emptyset \\
    T[(s+1,t),u-1,b] &\text{if}&  P\cap B^\uparrow(s,t)\neq\emptyset \\
    T[(s+1,t),u,b] &&  P\cap B^\downarrow(s,t)\neq\emptyset \end{array}\right\}$\\[1ex]
    $T[(s,t),u,b]=$&    \multicolumn{1}{c}{$\bigvee$}\\
    &$\left\{\begin{array}{l@{\quad}l@{\quad}l}
    T[(s,t+1),u,b]\wedge S[(s,t),u,b] && L\cap B^\rightarrow(s,t)\neq\emptyset \\
    T[(s,t+1),u,b-1] &\text{if}&  P\cap B^\rightarrow(s,t)\neq\emptyset \\
    T[(s,t+1),u,b] &&  P\cap B^\leftarrow(s,t)\neq\emptyset \end{array}\right\}$.
  \end{tabular}

\end{lemma}

\begin{proof}
  We show equivalence of the two terms. Let~$C$ be an $xy$-monotone curve
  from~$r$ to~$(s,t)$. Let~$e$ be the last segment of~$C$ and let~$C'=C-e$.
  Since~$C$ is $xy$-monotone, $C'$ ends either at the grid point $(s+1,t)$
  or at $(s,t+1)$. Without loss of generality, we assume that~$C'$
  ends at~$(s+1,t)$. We show that~$T[(s,t),u,b]=\true$ if and only if
  the first term of the right hand side is $\true$. Analogous arguments
  apply for~$C'$ ending at~$(s,t+1)$ and the second term. Note that, by
  construction, property (i) is satisfied for~$C$ and~$C'$.

  We distinguish cases based on whether the traversal along the segment~$e$
  from~$(s+1,t)$ to~$(s,t)$ is a port event or a site event.

  \textbf{Case 1:} Traversal of~$e$ is a port event.  Since~$e$ passes a
  port, all sites that lie in the rectangle spanned by~$r$ and~$G(s,t)$
  also lie in the rectangle spanned by~$r$ and~$G(s+1,t)$. Thus,
  the numbers~$u$ and~$b$ of such sites above and below~$C$ is the same as
  the numbers of sites above and below~$C'$, respectively. Hence, property
  (ii) holds for~$C$ if and only if it holds for~$C'$.

  Because~$C'$ is
  a subset of~$C$, the strip condition holds for every point of~$C$
  if and only if it holds for every point of~$C'$ and for~$(s,t)$.
  Thus, property (iii) is satisfied for~$C$ if and only if it is satisfied
  for~$C'$ and $S[(s,t),u,b]=\true$.

  \textbf{Case 2:} Traversal of~$e$ passes a site~$p$.
  For property (iii), observe that, since the traversal of~$e$ is a site
  event, the strip conditions for~$(s,t)$ and~$(s+1,t)$ are equivalent
  by Observation~\ref{obs:stripwalk}.

  For property (ii), note that, except for~$p$, the sites that
  lie in the rectangle spanned
  by~$r$ and~$G(s,t)$ also lie in the rectangle spanned by~$r$
  and~$G(s+1,t)$. If~$p$ lies above~$e$, there are~$u$ sites above
  and~$b$ sites below~$C$ if and only if there are~$u-1$ sites above
  and~$b$ sites below~$C'$, respectively. Symmetrically, if~$p$ lies
  below~$e$, there are~$u$ sites above and~$b$ sites below~$C$ if and only
  if there are~$u$ sites above and~$b-1$ sites below~$C'$, respectively.
  In either case, $C$ satisfies condition (ii) if and only if~$C'$ does.
\end{proof}

Clearly, the recurrence from Lemma~\ref{lem:recurrence} can be used to
compute~$T$ in polynomial time via dynamic programming. Note that it
suffices to store~$u$, as the number of sites below the curve~$C$ can
directly be derived from~$u$ and all sites that are contained in the
rectangle spanned by~$r$ and~$G(s,t)$. Thus, in the following we work
with~$T[(s,t),u]$. The running time crucially relies on the number of
strip conditions that need to be checked.  We show
that after a $O(n^2)$ preprocessing phase, such queries can be answered
in~$O(1)$ time.

To implement the test of the strip conditions, we use a table~$B_\TT$,
which stores in~$B_\TT[s,t]$ how large a
deficit of sites to the right can be compensated by sites above and
to the left of~$G(s,t)$.  That is,~$B_\TT[s,t]$ is the maximum
value~$k$ such that there exists a rectangle~$K_{B_\TT[s,t]}$ with
lower right corner~$G(s,t)$ whose top side is bounded by the top
side of~$R$, and that contains~$k$ more sites in its interior, than
it has ports on its top side.  Once we have computed this matrix, it
is possible to query the strip condition in the dynamic
program that computes~$T$ in~$O(1)$ time as follows:  Assume we have an
entry~$T[(s,t),u]$, and we wish to check its strip condition.
Consider a curve~$C$ from~$r$ to~$G(s,t)$ such that~$u$ sites are
above~$C$.  The strip condition is satisfied if and only
if~$u+B_\TT[s,t]$ is at least as large as the number of top ports to
the right of~$G(s,t)$. This is true if the rectangle spanned by the lower
left corner of~$K_{B_\TT[s,t]}$ and~$r$ contains at
least~$u+B_\TT[s,t]$ sites, which is an upper bound on the number of
ports on the top side of that rectangle.

We now show how to compute~$B_\TT$ in~$O(n^2)$ time.  We compute each
row separately, starting from the left side.  We initialize~$B_\TT[0,t]
= 0$ for~$t=0,\dots,M$, since in the final column, no deficit can
be compensated.  The matrix~$B$ can be filled by a horizontal sweep.
The entry~$B_\TT[s,t]$ can be derived from the already
computed entry~$B_\TT[s-1,t]$.  If the step from~$s-1$ to~$s$ is a site event,
 the amount of the deficit we can compensate increases by~1. If
it is a port event the amount of the deficit we can compensate decreases by~1.  Moreover, the
compensation potential never goes below~0.  We obtain
$$B_\TT[s,t]=
\begin{cases}
  B_\TT[s-1,t] + 1 & \text{if step is site event}, \\
  \max\{B_\TT[s-1,t] - 1,0\} & \text{if step is port event}. \\
\end{cases}
$$ The table can be clearly filled out in $O(n^2)$ time.  A similar
matrix~$B_\RR$ can be computed for the vertical strips.  Altogether,
this yields an algorithm for \textsc{Two-Sided Boundary Labeling with
  Adjacent Sides} that runs in~$O(n^3)$ time and uses~$O(n^3)$ space.
However, the entries of
each row and column of~$T$ depend only on the previous row and column,
which allows us to reduce the storage requirement to~$O(n^2)$.  Using
Hirschberg's algorithm~\cite{h-lsamcs-75}, we can still backtrack the
dynamic program and find a solution corresponding to an entry in the
last cell in the same running time.  We have the following theorem.

\newcommand{\thmTwoSidedCorrectText}{\textsc{Two-Sided Boundary
  Labeling with Adjacent Sides} can be solved in~$O(n^3)$ time
  using~$O(n^2)$ space.}
\begin{theorem}\label{thm:two-sided-correct}
  \thmTwoSidedCorrectText
\end{theorem}

Our next goal is to improve the performance of our algorithm by
reducing the number of dimensions of the table~$T$ by~1.  As a first
step, we show that for any search position~$\vect c=(s,t)$, the 
set of all~$u$ with~$T[\vect c, u] = \true$ is an
interval.

\begin{lemma}\label{lem:interval}
  Let~$T[\vect c, u] = T[\vect c,u'] = \true$  with~$u <
  u'$.  Then~$T[\vect c, u''] = \true$  for~$u \le u'' \le
  u'$.
\end{lemma}

\begin{pf}
  Let~$C$ be a curve corresponding to the entry~$T[\vect c, u]$.  That
  is~$C$ connects~$r$ to~$\vect c$ such that any point on~$C$
  satisfies the strip condition.  Similarly, let~$C'$ be a curve
  corresponding to~$T[\vect c, u']$; see
  Fig.~\ref{fig:storage_improvement}.

  Since~$u$ and~$u'$ differ, there is a rightmost site~$p$, such
  that~$p$ is below~$C$ and above~$C'$.  Let~$v$ and~$v'$ be the grid
  points of~$C$ and~$C'$ that are immediately to the left of~$p$.
  Note that~$v$ is above~$v'$ since~$C$ is above~$p$ and~$C'$ is below
  it.  Consider the curve~$C''$ that starts at~$r$ and follows~$C$
  until~$v$, then moves down vertically to~$v'$, and from there
  follows~$C'$ to~$p$.  Obviously~$C''$ is an $xy$-monotone curve, and
  it has above it the same sites as~$C'$, except for~$p$, which is
  below it.  Thus there are~$u'' = u'-1$ sites above~$C''$ in the
  rectangle spanned by~$p$ and~$r$.  If all points of~$C''$
  satisfy the strip condition, then this implies~$T[\vect c,
  u''] = \true$.

  We show that indeed the strip condition is satisfied for any point on~$C''$.
  Let~$C_1$ be the subcurve of $C''$ that connects $r$ to $v$, let $C_2$ be
  the segment $vv'$ and let~$C_3$ be the subcurve of~$C''$ that
  connects $v'$ to $\vect c$.
  Since $C_1$ is also a subcurve of~$C$ and it starts at $r$, it directly
  follows that any point of~$C_1$ satisfies the strip condition.
  For the
  points on~$C_2$ we can argue as follows. Since~$C_2$ lies below~$C$
  and any point of $C$ satisfies the horizontal strip condition, any point
  of~$C_2$ must satisfy the horizontal strip condition. Analogously, because
  $C_2$ lies above~$C'$ and any point of $C'$ satisfies the vertical strip
  condition, each point of $C_2$ must satisfy the vertical strip condition.
  Finally, since~$C_3$ is a subcurve of $C'$, any point of $C'$ satisfies the
  strip condition and any point of $C_1$ and $C_2$ satisfies the strip
  condition, it directly follows that any point of $C_3$ satisfies the strip
  condition.
\end{pf}

Using Lemma~\ref{lem:interval}, we can reduce the dimension of the
table~$T$ by~1.  It suffices to store at each entry~$T[\vect c]$ the
boundaries of the $u$-interval.  This reduces the amount of storage
to~$O(n^2)$ without increasing the running time.  Using Hirschberg's
algorithm, the storage for~$T$ even decreases to~$O(n)$.  Tables~$B_\TT$ and~$B_\RR$ 
still have size~$O(n^2)$, however.

Our next goal is to reduce the running time to~$O(n^2)$.  An entry
in~$B_\TT[s,t]$ tells us which deficits can be compensated.  This can
also be interpreted as a lower bound on the number of sites a curve
from~$r$ to~$G(s,t)$ must have above it, in order to satisfy the
horizontal strip condition.  Namely, let~$\tau_{s,t}$ denote the
number of ports on the top side of the rectangle spanned by~$G(s,t)$
and~$r$.  Then~$u \ge \tau_{s,t} - B_\TT[s,t]$ is equivalent to
satisfying the horizontal strip condition for the strip directly
above~$G(s,t)$.  Similarly, the corresponding entry~$B_\RR[s,t]$ gives
a lower bound on the number of sites below such a curve, which in
turn, together with the number of sites contained in the rectangle
spanned by~$G(s,t)$ and~$r$ implies an \emph{upper bound} on the
number of sites above the curve.  Thus,~$B_\TT$,~$B_\RR$, and the
information on how many sites, top ports and right ports are in the
rectangle spanned by~$G(s,t)$ and~$r$ together imply a lower and an
upper bound, and thus an interval of~$u$-values, for which the
horizontal and vertical strip conditions at~$G(s,t)$ is satisfied.
Hence the program can simply intersect this interval with the
union of the intervals obtained from~$T[(s,t) - \vect{\Delta c}]$,
where~$\vect{\Delta c}$ has exactly one non-zero entry, which is~1.
Consequently, the amount of work per entry of~$T$ is still~$O(1)$.
Note that by Lemma~\ref{lem:interval} the result of this computation
is again an interval.

\REMOVE{ Next, we would like to reduce the storage using Hirschberg's
algorithm~\cite{h-lsamcs-75}, which immediately reduces the storage
requirement of~$T$ to~$O(n)$.  We would like to reduce the storage
requirement for~$B_\TT$ and~$B_\RR$ by using Hirschberg's algorithm as
well.  However,}

Now we turn to the space consumption.  Hirschberg's algorithm~\cite{h-lsamcs-75} 
immediately reduces the space consumption of~$T$ to~$O(n)$.  We would
like to apply the same trick to~$B_\TT$ and to~$B_\RR$.  Recall that~$B_\TT$ is 
computed from left to right and~$B_\RR$ from bottom to top.
Unfortunately, this is opposite to the order we use for computing~$T$,
where we proceed from top-right to bottom-left. We can fix this problem by running
the dynamic programs for computing~$B_\TT$ and~$B_\RR$ backwards, by
precomputing the entries of~$B_\TT$ and~$B_\RR$ on the top and right side,
and then running the updates backwards.  This allows us to use Hirschberg's algorithm, 
and the algorithms can run in a synchronized manner such
that at any point in time the required data is available, using
only~$O(n)$ space.  

A new issue, however, appears.  The update~$B_\TT[s,t] = \max \{B_\TT[s-1,t]
- 1,0\}$ is not easily reversible.  When running the dynamic
program backwards, it is not clear whether~$B_\TT[s,t] = 0$
implies~$B_\TT[s-1,t] = 0$ or~$B_\TT[s-1,t] = 1$ at a port step.
To remedy this issue, fix a column~$s$ of the table
corresponding to a port event and consider the circumstances under
which~$B_\TT[s-1,t] - 1 = -1$, i.e.,~$B_\TT[s-1,t] = 0$.  This implies
that, for any rectangle~$K$ with lower right corner~$G(s-1,t)$ whose
top side is contained in the top side of~$R$, there are at most as
many sites in~$K$ as there are ports in the top side of~$K$.  Assume
that this is the case for some fixed value~$t_0$, i.e., ~$B_\TT[s-1,t_0]$.
Since the possible rectangles for an entry~$B_\TT[s-1,t]$ with~$t \ge
t_0$ contain at most as many sites as the ones for~$B_\TT[s-1,t_0]$,
this implies~$B_\TT[s-1,t_0] = B_\TT[s-1,t] = 0$ for all~$t \ge t_0$.  If
on the other hand,~$t_0$ is such that~$B_\TT[s-1,t_0] > 0$, then the
rectangles corresponding to~$B_\TT[s-1,t]$ for~$t<t_0$ contain at least as
many sites as the ones for~$B_\TT[s-1,t_0]$, and we have~$B_\TT[s-1,t] \ge
B_\TT[s-1,t_0]$ for~$t<t_0$.  Thus, there is a single gap~$t_0$
such that, for any~$t \ge t_0$, we have~$B_\TT[s-1,t] = 0$ and, for 
any~$t<t_0$, we have~$B_\TT[s-1,t] > 0$; see Fig.~\ref{fig:gap_rectangle}.
Storing this gap for each column~$s$ that is a port event allows us to
efficiently reverse the dynamic program.  Note that storing one value
per column only incurs~$O(n)$ space overhead.  Of course, the same
approach works for the dynamic program computing~$B_\RR$.  Overall, 
we have shown the following theorem.

\newcommand{\thmEfficientTwoSidedText}{\textsc{Two-Sided Boundary
  Labeling with Adjacent Sides} can be
  solved in $O(n^2)$ time using~$O(n)$ space.}
\begin{theorem}\label{thm:efficient-twosided-labeling}
  \thmEfficientTwoSidedText
\end{theorem}

\section{Extensions}\label{sec:two-sided-extensions}

The techniques we used to obtain
Theorem~\ref{thm:efficient-twosided-labeling} can be applied to
solve a variety of different extensions of the two-sided labeling
problem with adjacent sides.  We now show how to
\begin{inparaenum}[a)]
  \item generalize to sliding ports instead of fixed ports,
  \item maximize the number of labeled sites, and
  \item minimize the total leader length in a planar solution.
\end{inparaenum}

  \begin{figure}[tb]
    \begin{minipage}[b]{.25\textwidth}
    \centering
    \includegraphics[width=\linewidth]{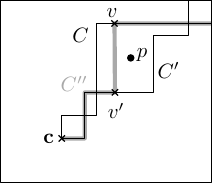}
   \end{minipage}\hfill
    \begin{minipage}[b]{.42\textwidth}
    \centering
    \includegraphics{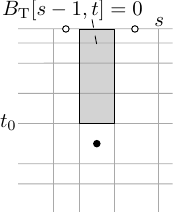}
   \end{minipage}\hfill
  \begin{minipage}[b]{.25\textwidth}
    \centering
   \includegraphics[page=1,width=\linewidth]{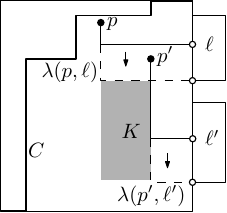}
  \end{minipage}

  \begin{minipage}[t]{.25\textwidth}
    \centering
    \caption{Sketch for the proof of Lemma~\ref{lem:interval}.}
    \label{fig:storage_improvement}
  \end{minipage}\hfill
  \begin{minipage}[t]{.42\textwidth}
    \centering
     \caption{The gap~$t_0$ is defined such that we have~$B_\TT[s-1,t] = 0$ for any~\mbox{$t \ge
         t_0$}, and~$B_\TT[s-1,t] > 0$ for any~\mbox{$t < t_0$}.}
     \label{fig:gap_rectangle}
  \end{minipage}\hfill
    \begin{minipage}[t]{.25\textwidth}
    \centering
    \caption{Sketch for the proof of Lemma~\ref{lem:sliding}}
   % \caption{The grid induced by sites and the boundary of the
%labels}
    \label{fig:sliding-ports}
    \end{minipage}
    %\belowfigure
\end{figure}

\subsection{Sliding Ports}\label{subsec:sliding}

First, observe that Proposition~\ref{prop:solut}, which guarantees the
existence of an $xy$-separated planar solution, also holds for sliding
ports.  The same proofs apply by conceptually fixing the ports of a
given planar solution when applying the rerouting operations.  The following
lemma shows that, without loss of generality, we can simply fix all
ports at the bottom-left corner of their corresponding labels.  This
immediately solves the problem.

\begin{lemma}\label{lem:sliding}
  If there exists an $xy$-separated planar
  solution~$\mathcal L$ for the two-sided boundary labeling problem with
  adjacent sides and sliding ports, then there also exists an $xy$-separated planar
  solution~$\mathcal L'$ in which the ports are fixed at the bottom
  left corners of the labels.
\end{lemma}

\begin{pf}
  We show how to transform~$\mathcal L$ into~$\mathcal L'$. Let~$C$ be
  the $xy$-monotone curve that separates the top leaders  from the right
  leaders of~$\mathcal L$. We move the ports induced by~$\mathcal L$ to
  the bottom-left corner of their corresponding labels such that the
  assignment between labels and sites remains;
  see~Fig.\ref{fig:sliding-ports}. Obviously, the bends of the leaders
  connected to the right side only move downwards. Thus, the leaders
  lie entirely below~$C$. Symmetrically, the bends of the leaders connected to
  the top side only move to the left and thus these leaders lie entirely
  above~$C$.

  Consequently, only conflicts between the same type of leaders can
  arise. Consider the topmost intersection of two
  leaders~$\lambda(p,\ell)$ and~$\lambda(p',\ell')$ connected to the
  right side and assume that~$p$ lies to the left of~$p'$. Let~$K$ be
  the rectangle that is spanned by the bends of~$\lambda(p,\ell)$
  and~$\lambda(p',\ell')$. Due to moving the ports downwards, the
  leaders lie entirely below~$C$ and the bend of~$\lambda(p',\ell')$ must lie
  below~$\lambda(p,\ell)$. Hence,~$K$ lies completely
  in~$R_\RR$. In order to resolve the conflict, we reroute~$p$ to~$\ell'$
  and~$p'$ to~$\ell$ using the bottom-left corners of~$\ell$ and~$\ell'$
  as ports. Obviously, the leaders only change on~$\partial K$. Therefore, new
  conflicts can only arise on the left and bottom sides of~$K$. In
  particular, only the leader of~$\ell'$ can be involved in new
  conflicts, while the leader of~$\ell$ is free of any conflict. Thus,
  after finitely many such steps we have resolved all
  conflicts, from top to bottom. Symmetric arguments apply for the leaders
  connected to the top side.
\end{pf}

\begin{theorem}\label{thm:twosided-sliding-ports-labeling}
  \textsc{Two-Sided Boundary Labeling with Adjacent Sides and Sliding
    Ports} can be solved in~$O(n^2)$ time using~$O(n)$ space.
\end{theorem}

\subsection{Maximizing the Number of Labeled Sites}\label{subsec:maximize}

So far our algorithm only returns a leader layout if there is a planar
solution that matches each label to a site. As Bekos et al.~\cite[Fig.~16]{bkps-afbl-10}
observed, this need not always be the case, so it becomes important
to be able to maximize the number of labels
connected to sites in a planar solution. We achieve this by
removing labels from a given instance and using our algorithm to
decide whether a crossing-free solution exists.

Lemma~\ref{lem:sliding} shows that we can move top ports to the left
and right ports to the bottom without making a solvable instance
unsolvable. Thus, it suffices to remove the rightmost top labels
and the topmost right labels. Let~$k$ be the number of labels we
want to use with~$k_\TT$ of them being top labels
and~$k_\RR$ right labels, so that~$k_\TT+k_\RR=k$. For a given~$k$, we can
decide whether a crossing-free solution that uses exactly~$k$ labels
exists by removing the~$m_\TT - k_\TT$ rightmost
top labels and the~$m_\RR - k_\RR$ topmost right
labels for any possible~$k_\TT$ and~$k_\RR$. We
therefore start with~$k_\TT = \min\{k,m_\TT\}$
and~$k_\RR=k-k_\TT$. We keep decreasing~$k_\TT$
and increasing~$k_\RR$ by~1, until a crossing-free solution
is found or~$k_\RR=\min\{k,m_\RR\}$. In the latter case,
no crossing-free solution that uses exactly~$k$ labels exists.
With this approach we can use binary search to find the maximum~$k$,
using our algorithm up to~$k$ times per step. Since~$k \le n$, this
yields an algorithm for \textsc{Two-Sided Boundary Labeling with
Adjacent Sides} that maximizes the number of labeled sites that
runs in~$O(n^3\log n)$ time and uses~$O(n)$ space.

\begin{theorem}
  \sloppy
  \textsc{Two-Sided Boundary Labeling with Adjacent Sides} can be solved
  in $O(n^3\log n)$ time using~$O(n)$ space such that the number of labeled
  sites is maximized.
\end{theorem}

\ADD{ Assume that~$t$ sites cannot be labeled. Then \textsc{Two-Sided
    Boundary Labeling with Adjacent Sides} can be solved in $O(n^2
  t\log t)$ time using~$O(n)$ space and such that the number of labeled
  sites is maximized. To that end we use the following approach. We
  check for $h=2^i$ with~$i=0,1,2,\dots$ whether there is a
  planar solution with~$h$ unlabeled sites. We stop this procedure
  when we have found such a solution, which takes place
  after~$\lceil \log(t) \rceil$ steps. Using the approach described above, we
  need~$O(n^2t)$ time for each test. We then know that~$\frac{h}{2}<t
  \leq h$. We apply a binary search to determine~$t$. Overall, this
  approach needs~$O(n^2 t\log t)$ time.}

\subsection{Minimizing the total leader length}\label{subsec:minimize}

Recall that, by Proposition~\ref{prop:solut}, there always exists a
length-minimal planar solution that is $xy$-separated. To
obtain a length-minimal planar solution, we mainly change the
table~$T$ used by the dynamic program given in
Section~\ref{sec:two-sided-algorithm}.  Let~$C$ be an $xy$-monotone
curve~$C$ that starts at~$r$ and ends at~$G(s,t)$. We assign to every
table entry the length of the leaders that are connected to the ports
in the rectangle~$K$ spanned by~$r$ and~$G(s,t)$.

\begin{figure}[tb]
  \centering
  \begin{subfigure}{.45\textwidth}
    \centering
    \includegraphics[page=1]{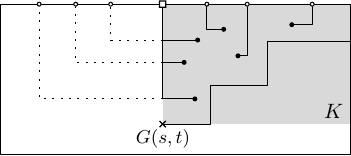}
    \caption{}\label{fig:leader_length_min:1}
  \end{subfigure}
  \begin{subfigure}{.45\textwidth}
    \centering
    \includegraphics[page=2]{leader_length_min}
    \caption{}\label{fig:leader_length_min:2}
  \end{subfigure}
  \caption{\ADD{Illustration of the curve~$C$ and the rectangle $K$
      spanned by~$G(s,t)$ and the top-right corner of~$R$.
      \subref{fig:leader_length_min:1}) There are more sites than
      ports in~$K$ above~$C$. The unlabeled sites are connected to a
      dummy port located at the top-left corner of~$K$. The dummy port
      is illustrated as a square.  \subref{fig:leader_length_min:2})
      There are more ports than sites in~$K$ above~$C$. The unlabeled
      ports are labeled to sites that lie to the left of~$K$, which
      induce the front with bottom-left point~$\vect F_\TT$.  }}
  \label{fig:leader_length_min}
\end{figure}

If there are more sites than top ports in~$K$ above~$C$, we have to
connect some of these sites to ports that lie to the left of~$K$\ADD{; see
Fig.~\ref{fig:leader_length_min:1}}. The vertical lengths of their
leaders, however, are fixed. We imagine a dummy top port at the left
border of~$K$ and connect all unlabeled sites to this port. When
traversing the grid horizontally, this dummy port moves to the
left. In order to update the total length of the leaders in~$K$, we
only have to keep track of the number of unlabeled sites and increase
the horizontal length of their leaders. The sites in~$K$ below~$C$ are
handled analogously.

If there are more top ports than sites in~$K$ above~$C$,
we have to connect these ports to sites that lie to the left of~$K$\ADD{; see Fig.~\ref{fig:leader_length_min:2}}.
In order to remember which sites are already labeled, we store the
\emph{top front} as the rectangle with top-right corner~$r$ that
includes all sites that are already connected to a top port inside~$K$,
and the \emph{right front} as the rectangle with top-right corner~$r$ that
includes all sites that are already connected to a right port inside~$K$.

Let~$\vect F_\TT=(x_\TT,y_\TT)$ be the bottom-left point of the top
front for a given $xy$-monotone curve~$C$ that starts at~$r$ and ends
at~$G(s,t)$. Similarly, let~$\vect F_\RR=(x_\RR,y_\RR)$ be the right
front for~$C$. We define $T[\vect c=(s,t),u,\vect F_\TT,\vect
F_\RR]=(l,g_\TT,g_\RR)$ if there exists an $xy$-monotone curve~$C$ and
leaders inside $K \cup F_\TT \cup F_\RR$ such that the following
conditions hold, otherwise it contains~$(-1,0,0)$.
\begin{compactenum}[(i)]
  \item Curve~$C$ starts at the top-right corner~$r$ of~$R$ and ends
        at~$G(s,t)$.
  \item Inside the rectangle~$K$ spanned by~$r$
        and~$G(s,t)$, there are~$u$ sites of~$P$ above~$C$.
  \item For each strip in the two regions~$R_\TT$ and~$R_\RR$ defined
        by~$C$ the strip condition holds.
  \item The sites in~$K\cup \vect F_\TT\cup \vect F_\RR$ are connected to the
  		  ports on the border of~$K\cup \vect F_\TT\cup \vect F_\RR$ such that the
  		  induced solution is	planar, length-minimal, the sites
  		  above~$C$ or in~$\vect F_\TT$ are only connected to top ports, and
  		  the sites below~$C$ or in~$\vect F_\RR$ are only connected to right
  		  ports.
	\item There are~$g_\TT$ unlabeled top sites and~$g_\RR$ unlabeled
        right sites in~$K$.
\end{compactenum}

Note that~$g_\TT$ and~$g_\RR$ depend on~$s$,~$t$,~$u$ and
can be precomputed, but to make the algorithm more intuitive, we update these values
on-line and store them in~$T$.
We first describe how to handle the top front while traversing the
grid. Initially,~$\vect F_\TT=G(s,t)$. As long as we have more top sites
than top ports in~$K$, we can connect all ports to sites and thus can
maintain~$\vect F_\TT=G(s,t)$. Once we have exactly the same number
of top ports and top sites in~$K$ and we encounter a port event for a top
port, we have to check the strip condition and find the
rightmost point~$F_\TT$ with~$y(\vect F_\TT)=G_y(t)$ such that the rectangle~
$R_{\vect F_\TT}$ spanned by~$\vect F_\TT$ and~$r$ is valid. By storing~$\vect F_\TT$, we
know that all ports to the right of~$x(\vect F_\TT)$ are already connected
to a site, all sites to the top-right of~$\vect F_\TT$ are already
connected to a port, and all top sites to the bottom-left of~$\vect F_\TT$
have to be connected to a port that lies to the left of~$x(\vect F_\TT)$. Thus, we do
not have to check new strip conditions until~$s<x(\vect F_\TT)$. We
handle~$\vect F_\RR$ similarly.

We now look at the length of the top leaders, the length of the
right leaders can be handled similarly. Note that by moving from~$t$
to~$t-1$, the length of the top leaders does not change. If~$g_\TT>0$,
we imagine an additional port at~$x(\vect F_\TT)$ that can be connected to
~$g_\TT$ top sites.  When moving from~$s$ to~$s-1$, we add $g_\TT\cdot(G
_x(s)-G_x(s-1))$ to~$l$.  When we calculate a new value~$\vect F_\TT$ by
checking the strip condition, we can immediately connect all
top sites inside the top front to top ports, and add the
corresponding leader length to~$l$. Thus, we only encounter
site events for sites that are a) inside $\vect F_\TT \setminus K$ or b)
have to be connected to a top port that lies to the left of~$x(\vect F_\TT)$. In case a) we
do not change~$l$, in case b) we connect the site to the imaginary
port, add the length of the corresponding leader to~$l$ and
increase~$g_\TT$ by 1. When we encounter a port event, if the port
lies inside~$\vect F_\TT \setminus K$, we do not change~$l$, otherwise we
can connect any of the unlabeled sites to this port. We add the
horizontal distance between~$G_x(s)$ and the port to~$l$ and
decrease~$g_\TT$ by 1. Note that by choosing any unlabeled site, the
resulting solution may not be planar. However, because the bends of
all unconnected sites will be above~$C$, we can use
Lemma~\ref{lem:boxlemma} to remove the crossings without changing the
total leader length.

Since the matrix now has four additional fields, the running time and
storage is increased by a factor of $n^4$ over the algorithm from
Theorem~\ref{thm:two-sided-correct}.  Additionally, we need~$O(n\log
n)$ time to check the strip condition and to compute a length-minimal
solution for the sites and ports inside~$\vect F_\TT \setminus K$ and~$\vect F_\RR
\setminus K$.

\begin{theorem} \sloppy
  \textsc{Two-Sided Boundary Labeling with Adjacent Sides} can be solved
  in $O(n^8\log n)$ time using~$O(n^6)$ space such that the total
  leader length is minimized.
\end{theorem}

Using an appropriate data structure to precompute the fronts, it may
be possible to decrease the running time slightly.

\section{The Three- and Four-Sided Cases}

In this section, we also allow labels on the bottom
and the left side of~$R$.  In order to solve an instance of the three-
and four-sided case, we adapt the techniques we developed for the
two-sided case.  We assume that the ports are fixed and the number of labels and sites is equal. 
In Section~\ref{sec:four-sided-solutions} we first analyze the structure
of planar solutions obtaining a result similar to
Proposition~\ref{prop:solut}.  In
Sections~\ref{sec:algorithm-three-sided-case}
and~\ref{sec:algorithm-four-sided-case}, we present
algorithms for the three- and four-sided cases.

\subsection{Structure of Three- and Four-Sided Planar Solutions}
\label{sec:four-sided-solutions}

Similar to our approach to two-sided boundary labeling, we pursue the idea that if there exists a planar solution, then we can also find a planar solution such that there are four $xy$-monotone curves connecting the
four corners of~$R$ to a common point~$o$, and such that these curves
separate the leaders of the different label types from each other; see
Fig.~\ref{fig:partitions}.
\begin{figure}[tb]
  \centering
  \includegraphics[page=2]{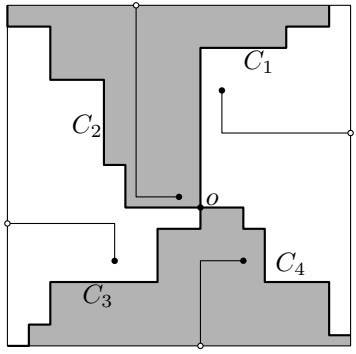}
  \caption{The curves $C_1$, $C_2$, $C_3$ and $C_4$ meeting at the
     point $o$ partition the rectangle into four regions.}
  \label{fig:partitions}
\end{figure}
To that end, we first show that leaders of left and right labels can
be separated vertically and leaders of top and bottom labels can be
separated horizontally.  Afterwards, we apply the result of
Lemma~\ref{lem:xysep} in order to resolve the remaining overlaps,
e.g., between top and right leaders. We first introduce some
notions. 

\begin{definition}\label{def:xsep}
  A planar solution for the four-sided boundary labeling problem is
  \begin{compactenum}[(i)]
  \item \emph{$x$-separated} if there exists a vertical line~$\ell$
    such that the sites that are labeled to the left side are to the
    left of~$\ell$ and the sites that are labeled to the right side
    are to the right of~$\ell$, and
  \item \emph{$y$-separated} if there exists a horizontal line~$\ell$
    such that the sites that are labeled to the top side are
    above~$\ell$ and the sites that are labeled to the bottom side are
    below~$\ell$.
  \end{compactenum}
\end{definition}

A left leader~$\lambda$ and a right leader~$\lambda'$ \emph{overlap}
if $x(\bend(\lambda)) > x(\bend(\lambda'))$. Analogously, a bottom
leader~$\lambda$ and a top leader~$\lambda'$ \emph{overlap} if
$y(\bend(\lambda)) > y(\bend(\lambda'))$. Hence, a planar
solution~$\mathcal L$ is both~$x$-separated and~$y$-separated if and
only if no left and right leaders overlap, and no bottom and top
leaders overlap.  We are now ready to prove that we can always
find a planar solution that is both~$x$-separated and~$y$-separated,
if a solution exists.

\begin{figure}[tb]
  \centering
  \begin{subfigure}{.45\textwidth}
    \centering
    \includegraphics[page=1]{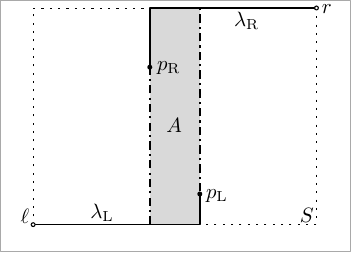}
    \caption{}\label{fig:resolve_overlap:1}
  \end{subfigure}
  \begin{subfigure}{.45\textwidth}
    \centering
    \includegraphics[page=2]{resolving_horizontal_overlap_new}
    \caption{Case 1.}\label{fig:resolve_overlap:2a}
  \end{subfigure}
  \begin{subfigure}{.45\textwidth}
    \centering
    \includegraphics[page=3]{resolving_horizontal_overlap_new}
    \caption{Case 2.1}\label{fig:resolve_overlap:2b}
  \end{subfigure}
  \begin{subfigure}{.45\textwidth}
    \centering
    \includegraphics[page=4]{resolving_horizontal_overlap_new}
    \caption{Case 2.2}\label{fig:resolve_overlap:2c}
  \end{subfigure}
  \caption{Different constellations of leaders intersecting the
    rectangle~$A$.  \subref{fig:resolve_overlap:1}) The rectangle~$A$
    is empty.
    \subref{fig:resolve_overlap:2a})--\subref{fig:resolve_overlap:2c})~Different
    cases where~$A$ is intersected by a top leader. $A$ is not
    explicitly illustrated, but spanned by $\bend(\lambda_\RR)$ and
    $\bend(\lambda_\LL)$.}
  \label{fig:resolve_overlap}
\end{figure}

\begin{lemma}\label{lem:xsep}
\ADD{If there exists a planar solution for the four-sided
  boundary labeling problem, then there exists a planar
  solution~$\mathcal L$ that is both $x$-separated and~$y$-separated.}
\end{lemma}

\begin{pf}
 % \color{blue}
  Among all planar solutions let $\mathcal L$ be one that minimizes
  $|\mathcal L|_x+|\mathcal L|_y$. We prove that then $\mathcal L$ is
  $x$- and $y$-separated by showing that otherwise we could reroute
  some leaders and obtain a planar solution~$\mathcal L'$ with
  $|\mathcal L'|_x+|\mathcal L'|_y < |\mathcal L|_x+|\mathcal L|_y$.

  Assume that~$\mathcal L$ is not $x$-separated. Symmetric arguments
  hold for the case that~$\mathcal L$ is not $y$-separated.  Then
  there exist sites~$p_\RR$ and~$p_\LL$ with~$x(p_\RR) < x(p_\LL)$,
  such that~$p_\RR$ is labeled by a right port~$r$, and~$p_\LL$ is
  labeled by a left port~$\ell$; see
  Fig.~\ref{fig:resolve_overlap:1}. Without loss of generality, assume
  that the horizontal segment of~$\lambda_\RR=\lambda(p_\RR,r)$ is
  above~the horizontal segment of~$\lambda_\LL=\lambda(p_\LL,\ell)$,
  otherwise we mirror the instance vertically.

  We choose $p_\LL$ and $p_\RR$ as a closest pair in the sense that
  the horizontal segments of their leaders have minimum vertical
  distance among all such pairs. Let~$A$ be the rectangle spanned
  by~$\bend(\lambda_\LL)$ and $\bend(\lambda_\RR)$.  By the minimality
  of $p_\LL$ and $p_\RR$, that rectangle can only be intersected by
  top and bottom leader, but not by left or right leaders. If no such
  leader intersects~$A$, we reroute~$p_\RR$ to the port
  of~$\lambda_\LL$ and~$p_\LL$ to the port of~$\lambda_\RR$, which
  decreases~$|\mathcal L|_x$ without increasing~$|\mathcal L|_y$; see
  Fig.~\ref{fig:resolve_overlap:1}. It does not introduce any
  crossings.

  In the following we assume that some leaders intersect~$A$.  Without
  loss of generality we assume that there is a top
  leader~$\lambda_\TT$ that intersects~$A$; otherwise we rotate the
  instance by $180^\circ$.  We denote its site by~$p_\TT$. Let~$S$ be
  the rectangle spanned by the ports~$\ell$ and~$r$; see
  Fig.~\ref{fig:resolve_overlap:1}. Depending on the leaders
  intersecting~$S$, we distinguish two cases.  Note that in
  particular~$\lambda_\TT$ intersects~$S$.

  \textbf{Case 1:} For any top leader~$\lambda$ intersecting~$S$ and
  for any bottom leader~$\lambda'$ intersecting~$S$ such
  that~$\lambda$ and~$\lambda'$ overlap, the site of~$\lambda$ lies to
  the left of the site of~$\lambda'$; see
  Fig.~\ref{fig:resolve_overlap:2a}. Let~$q_\RR$ denote the bottommost
  site that is connected by a right leader, and that lies in the
  rectangle spanned by~$\bend(\lambda_\TT)$ and $p_\RR$. Since~$p_\RR$
  lies in that rectangle, the site~$q_\RR$ exists. We denote the
  leader of~$q_\RR$ by~$\lambda'_\RR$. Further, let~$q_\TT$ be
  the topmost site that is connected by a top leader and that lies in
  the rectangle spanned by~$\bend(\lambda'_\RR)$ and the bottom-right
  corner of $R$. Since~$p_\TT$ lies in that rectangle, the
  site~$q_\TT$ exists. We denote its leader by~$\lambda'_\TT$.

  We now define two rectangles that we use to reroute leaders such
  that~$|\mathcal L|_x+|\mathcal L|_y$ is decreased and arising
  crossings can be resolved. The rectangle~$K_1$ is spanned
  by~$\bend(\lambda'_\RR)$ and $q_\TT$, and the rectangle~$K_2$ is
  spanned by~$\bend(\lambda'_\TT)$ and $q_\RR$.
  
  \begin{claim}\ 
  \begin{compactenum}[(1)]
  \item $K_1$ is only intersected by right leaders whose bends are
    contained in~$K_1$,
    \item $K_2$ is only intersected by top leaders whose bends are
      contained in~$K_2$, and
    \item $K_1$ and $K_2$ are internally disjoint.
  \end{compactenum}
  \end{claim}
   
  Assuming that the claim holds, we can reroute the sites as follows;
  we illustrate this rerouting by dash-dotted lines in
  Fig.~\ref{fig:resolve_overlap:2a}. The site~$q_\TT$ is rerouted to
  the port of~$\lambda'_\RR$ creating crossings only on the right side
  of~$K_1$. The site~$q_\RR$ is rerouted to the port of~$\lambda'_\TT$
  creating crossings only on the top side of~$K_2$.  Each rerouting
  decreases either~$\mathcal |\mathcal L|_x$ or $|\mathcal L_y|$
  increasing the other one. Further, only crossings between leaders of
  the same type are created. We apply Lemma~\ref{lem:boxlemma} to
  resolve the conflicts without increasing $|\mathcal L|_x$ or
  $|\mathcal L_y|$. In the remainder of this case we show that the
  stated claim holds.
  
  First, we show that~$K_1$ is only intersected by right leaders whose
  bends lie in~$K_1$. It is not intersected by any bottom leader,
  because such a leader would overlap~$\lambda'_\TT$, and its site
  would lie to the left of~$q_\TT$---a contradiction to the assumption
  of this case. It is not intersected by any left leader, because such
  a leader would intersect~$\lambda'_\TT$. It is not intersected by
  any top leader, because such a leader would either
  intersect~$\lambda'_\RR$ or contradict the choice
  of~$\lambda'_\TT$. Hence,~$K_1$ can only be intersected by right
  leaders. Further, all those leaders have their bend in~$K_1$,
  because the bottom-right corner is a site connected by a top
  leader. That leader would be intersected if a right leader
  intersecting~$K_1$ had its bend outside of~$K_1$.
 
  Next, we show that~$K_2$ is only intersected by top leaders whose
  bends lie in~$K_2$.  It is not intersected by any
  right leader, because such a leader would contradict the choice
  of~$\lambda'_\RR$ or intersect~$\lambda_\TT$. It is not intersected
  by any bottom leader, because such a leader would
  overlap~$\lambda'_\TT$, and its site would lie to the left
  of~$q_\TT$---a contradiction to the assumption of this case. It is
  not intersected by any left leader, because such a leader would
  intersect~$\lambda'_\TT$. Hence,~$K_2$ can only be intersected by
  top leaders. Further, all those leaders have their bend in~$K_2$,
  because the top-right corner is a site connected by a right leader.
 
  Finally, the rectangles~$K_1$ and~$K_2$ are internally disjoint, because~$K_1$ lies to
  the right of the vertical line through~$q_\RR$, while $K_2$ lies
  to the left of that line. 

  \textbf{Case 2:} There exist a top leader~$\lambda_\TT$
  intersecting~$S$ and a bottom leader~$\lambda_\BB$ intersecting~$S$
  such that they overlap and the site of~$\lambda_\TT$ lies to the
  right of the site of~$\lambda_\BB$; see
  Fig.~\ref{fig:resolve_overlap:2b}. Among
  all such pairs we choose~$\lambda_\TT$ and $\lambda_\BB$ such that
  their horizontal segments have minimal vertical distance. We denote
  the site of~$\lambda_\TT$ by~$p_\TT$ and the site of~$\lambda_\BB$
  by $p_\BB$. Due to the choice of~$\lambda_\TT$ and $\lambda_\BB$,
  the open rectangle that is spanned by~$p_\BB$ and~$p_\TT$ is
  intersected by no leader. The open rectangle spanned
  by~$\bend(\lambda_\TT)$ and $\bend(\lambda_\BB)$ is denoted by~$B$.
  Depending on the sites that are contained in~$B$, we distinguish
  four cases.

  \textbf{Case 2.1:} The rectangle $B$ contains no sites that are
  connected by left of right leaders; see
  Fig.~\ref{fig:resolve_overlap:2b}. Let~$K_1$ be the rectangle
  spanned by~$\bend(\lambda_\TT)$ and $p_\BB$, and let~$K_2$ be the
  rectangle spanned by~$\bend(\lambda_\BB)$ and $p_\TT$. While~$K_1$
  is only intersected by left leaders,~$K_2$ is only intersected by
  right leaders. Further, both rectangles are disjoint. We
  reroute~$p_\BB$ to the port of~$\lambda_\TT$ and $p_\TT$ to the port
  of~$\lambda_\BB$. Obviously, this decreases~$|\mathcal L|_y$ without
  increasing~$|\mathcal L|_x$. By applying Lemma~\ref{lem:boxlemma}, we
  resolve the arising conflicts.

  \textbf{Case 2.2:} The rectangle $B$ contains sites that are
  connected by left leaders as well as sites that are connected by right
  leaders; see Fig.~\ref{fig:resolve_overlap:2c}. Let~$q_\RR$ be the
  bottommost site in~$B$ that is connected to the right. We denote the
  leader of~$q_\RR$ by~$\lambda'_\RR$.  Let~$q_\BB$ be the leftmost
  site with~$y(q_\BB)\geq y(p_\BB)$ and $x(q_\BB)\leq x(p_\BB)$ that
  is connected to the bottom. Since~$p_\BB$ also satisfies these
  requirements, the site~$q_\BB$ exists. We denote the leader
  of~$q_\BB$ by~$\lambda'_\BB$. Let~$q_\LL$ be the topmost site in~$B$
  that is connected to the left. We denote the leader of~$q_\LL$
  by~$\lambda'_\LL$. Finally, let~$q_\TT$ be the rightmost site
  with~$y(q_\TT)\leq y(p_\TT)$ and $x(q_\TT)\geq x(p_\TT)$ that is
  connected to the top. Since~$p_\TT$ also satisfies these
  requirements, the site~$q_\TT$ exists. We denote the leader
  of~$q_\TT$ by~$\lambda_\TT'$.

  We now define four rectangles that we use to reroute leaders such
  that~$|\mathcal L|_x+|\mathcal L|_y$ is decreased and arising crossings
  can be resolved. The rectangle~$K_1$ is spanned by~$\bend(\lambda'_\RR)$
  and~$q_\BB$, the rectangle~$K_2$ is spanned
  by~$\bend(\lambda'_\BB)$ and~$q_\LL$, the rectangle~$K_3$ is spanned
  by~$\bend(\lambda'_\LL)$ and~$q_\TT$, and the rectangle~$K_4$ is
  spanned by~$\bend(\lambda'_\TT)$ and $q_\RR$.  Note that the
  rectangles~$K_3$ and~$K_4$ are rotationally symmetric to~$K_1$
  and~$K_2$, respectively. 
  
  \begin{claim}\ 
    \begin{compactenum}[(1)]
  \item $K_1$ is only intersected by right leaders whose bends are
    contained in~$K_1$,
  \item $K_2$ is only intersected by bottom leaders whose bends are
    contained in~$K_2$,
  \item $K_3$ is only intersected by left leaders whose bends are
    contained in~$K_3$,
  \item $K_4$ is only intersected by top leaders whose bends are
    contained in~$K_4$, and
    \item $K_1$, $K_2$, $K_3$ and $K_4$ are pairwise internally
      disjoint.
  \end{compactenum} 
  \end{claim}
  Assuming that the claim holds, we can reroute the sites in a
  circular fashion as follows; we illustrate the rerouting as
  dash-dotted lines in Fig.~\ref{fig:resolve_overlap:2c}. The
  site~$q_\BB$ is rerouted to the port of~$\lambda'_\RR$ creating
  crossings only on the right side of~$K_1$. The site~$q_\LL$ is
  rerouted to the port of~$\lambda'_\BB$ creating crossings only on
  the bottom side of~$K_2$. The site~$q_\TT$ is rerouted to the port
  of~$\lambda'_\LL$ creating crossing only on the left side
  of~$K_3$. Finally, the site~$q_\RR$ is rerouted to the port
  of~$\lambda'_\TT$ creating crossings only on the top side of~$K_4$.
  Each rerouting decreases either~$|\mathcal L|_x$ or $|\mathcal
  L_y|$ without increasing the other one. Further, only
  crossings between leaders of the same type are created. We apply
  Lemma~\ref{lem:boxlemma} to resolve the conflicts.  In the remainder
  of this case we show that the stated claim holds.
  
  First, we show that~$K_1$ is only intersected by right leaders whose
  bends lie in~$K_1$. This rectangle is not intersected by any bottom
  leader, because~$q_\BB$ is the leftmost site with~$y(q_\BB)\geq
  y(p_\BB)$ and $x(q_\BB)\leq x(p_\BB)$ that is connected to the
  bottom.  It is not intersected by any top leader, because such a
  leader would intersect~$\lambda'_R$ whose site lies
  below~$q_\BB$. Finally, it is not intersected by any left leader,
  because such a leader would intersect~$\lambda'_\TT$ whose site lies
  to the right of~$q_\BB$. Hence, only right leaders
  intersect~$K_1$. In particular, all those leaders have their bend
  in~$K_1$, because the bottom-right corner of~$K_1$ is the site of a
  bottom leader. That leader would be intersected if a right leader
  intersecting~$K_1$ had its bend outside of~$K_1$. Since $K_3$ is
  rotationally symmetric to~$K_1$, we can use symmetric arguments to
  prove that $K_3$ is only intersected by left leaders whose bends are
  contained in~$K_3$

  Next, we show that~$K_2$ is only intersected by bottom leaders whose
  bends lie in~$K_2$.  This rectangle is not intersected by any left
  leader, because such a leader would contradict the choice
  of~$q_\LL$. It is also not intersected by any top leader, because
  such a leader would intersect~$\lambda'_\RR$ or contradict the
  choice of~$\lambda_\TT$ and~$\lambda_\BB$. Finally, it cannot be
  intersected by any right leader, because such a leader would
  intersect~$\lambda'_\BB$. Hence, $K_2$ is only intersected by bottom
  leaders. Further, all those leaders have their bend in~$K_2$,
  because the bottom-left corner of~$K_2$ is a site connected to a
  left leader. That leader would be intersected if a bottom leader
  intersecting~$K_2$ had its bend outside of~$K_2$.  Since $K_4$ is
  rotationally symmetric to~$K_2$, we can use symmetric arguments to
  prove that $K_4$ is only intersected by top leaders whose bends are
  contained in~$K_4$.

  % The rectangle~$K_3$ is spanned by~$\bend(\lambda'_\LL)$
  % and~$q_\TT$. That rectangle is not intersected by any top leader,
  % because such a leader would contradict the choice of~$q_\TT$. It is
  % not intersected by any right leader, because such a leader would
  % intersect~$\lambda'_\BB$. It is not intersected by any bottom
  % leader, because such a leader would intersect~$\lambda'_\LL$, whose
  % site lies above~$q_\TT$. Hence, $K_3$ is only intersected by left
  % leaders. Further, all those leaders have their bend in~$K_3$,
  % because the top-left corner of~$K_3$ is a site connected by a top
  % leader. That leader would be intersected if a left leader
  % intersecting~$K_3$ has its bend outside of~$K_3$.

  % The rectangle~$K_4$ is spanned by~$\bend(\lambda'_\TT)$ and
  % $q_\RR$. That rectangle is not intersected by any right leader,
  % because such a leader would contradict the choice of~$q_\RR$. It is
  % not intersected by any bottom leader, because such a leader would
  % intersect~$\lambda_\LL$ or contradict the choice of~$\lambda_\BB$
  % and~$\lambda_\TT$. Finally, it is not intersected by any left
  % leader, because such a leader would
  % intersect~$\lambda'_\TT$. Hence,~$K_4$ is only intersected by top
  % leaders. Further, all those leaders have their bend in~$K_4$,
  % because the top-right corner is a site connected by a right
  % leader. That leader would be intersected if a top leader
  % intersecting~$K_4$ has its bend outside of~$K_4$.

  Finally, we show that the rectangles~$K_1$, $K_2$, $K_3$ and $K_4$ are pairwise
  internally disjoint. For a site~$p$
  let~$v(p)$ denote the vertical line through~$p$ and let~$h(p)$
  denote the horizontal line through~$p$. By construction we have
  that~$h(q_\BB)$ lies above~$h(q_\TT)$,~$K_1$ lies above~$h(q_\BB)$,
  and~$K_3$ lies below~$h(q_\TT)$. Hence, the rectangles~$K_1$ and
  $K_3$ are internally disjoint. Analogously, we have that~$v(q_\LL)$
  lies to the right of~$v(q_\RR)$,~$K_2$ lies to the right
  of~$v(q_\LL)$, and~$K_4$ lies to the left of~$v(q_\RR)$. Hence, the
  rectangles~$K_2$ and $K_4$ are internally disjoint. Further, the
  sites~$q_\LL$ and~$q_\RR$ lie in between~$h(q_\BB)$ and $h(q_\TT)$,
  because both lie in~$B$. Consequently,~$K_1$ and~$K_3$ do not
  intersect $K_2$ and~$K_4$, respectively.

  % We reroute the sites in a circular fashion. The site~$q_\BB$ is
  % rerouted to the port of~$\lambda'_\RR$ creating crossings only on
  % the right side of~$K_1$. The site~$q_\LL$ is rerouted to the port
  % of~$\lambda'_\BB$ creating crossings only on the bottom side
  % of~$K_2$. The site~$q_\TT$ is rerouted to the port of~$\lambda'_\LL$
  % creating crossing only on the left side of~$K_3$. Finally, the
  % site~$q_\RR$ is rerouted to the port of~$\lambda'_\TT$ creating
  % crossings only on the top side of~$K_4$.  Each rerouting effects
  % that either~$|\mathcal L|_x$ or $|\mathcal L_y|$ is decreased
  % without increasing the other one. Further, only crossings between
  % leaders of the same type are created. We therefore apply
  % Lemma~\ref{lem:boxlemma} to resolve the conflicts.

  \textbf{Case 2.3:} The rectangle $B$ contains only sites connected
  by right leaders. We apply the same procedure as in the previous
  case. However, we do not need to consider left leaders. Hence,~$K_3$
  is removed and $K_2$ is the rectangle that is spanned
  by~$\bend(\lambda'_\BB)$ and~$p_\TT$. By the choice of~$B$, the
  rectangle~$K_2$ is only intersected by right leaders whose bend is
  contained in~$K_2$. Further, the remaining rectangles~$K_1$, $K_2$ and
  $K_4$ are pairwise internally disjoint. The reroutings are again
  done in a circular fashion decreasing $|\mathcal L|_x+|\mathcal
  L|_y$. Finally, we apply Lemma~\ref{lem:boxlemma} to resolve
  crossings.

  \textbf{Case 2.4:} The rectangle $B$ contains only sites connected
  by left leaders. This case can be handled analogously to the
  previous case by mirroring the instance vertically.
\end{pf}

This lemma shows that, when searching for a planar solution of the
labeling problem, we can restrict ourselves to solutions that are
$x$-separated and~$y$-separated.  Let~$\mathcal L$ denote such a
solution, and let~$\ell_v$ and~$\ell_h$ be the lines separating the
sites labeled by left and right labels, and the ones labeled by top
and bottom labels, respectively.  Let~$o \in R$ denote the
intersection of~$\ell_v$ and~$\ell_h$, called \emph{center point}.
Let~$r_1,\dots,r_4$ denote the
corners of~$R$, named in counterclockwise ordering, and such
that~$r_1$ is the top-right corner. Consider the rectangles that are
spanned by~$o$ and~$r_i$ for~$i=1,\ldots,4$. Each of them contains
only two types of leaders. For example, the top-right rectangle
contains only top and right leaders.
An~$x$- and~$y$-separated planar solution is
\emph{partitioned} if, for each rectangle spanned by~$o$ and one of
the corners~$r_i$ of~$R$, there exists an
$xy$-monotone curve~$C_i$ from~$r_i$ to~$o$ that separates the two
different types of leaders contained in that rectangle; see
Fig.~\ref{fig:partitions}. Our next step is to show
that a planar solution can be transformed into a partitioned solution
without increasing~$|\mathcal L|_x$ and~$|\mathcal L|_y$.

\begin{proposition}\label{prop:partitioned}
  If there exists a planar solution~$\mathcal{L}$ for
  \textsc{Four-Sided Boundary Labeling}, then there exists a
  partitioned solution~$\mathcal{L}'$.
\end{proposition}

\begin{pf}
  By Lemma~\ref{lem:xsep}, we can assume that~$\mathcal L$ is $x$-
  and~$y$-separated.  Let~$o$ be the center point as defined above and
  let~$\ell_v$ be the vertical line through~$o$. We show how to ensure
  that the area~$K$ of~$R$ right of $\ell_v$ admits an $xy$-monotone
  curve from the top-right corner of~$R$ to $o$ that separates the top
  leaders from the right leaders inside $K$. The remaining cases are
  symmetric.

  Essentially, we proceed as in the proof of
  Proposition~\ref{prop:solut} to remove obstructions of
  types~(\subref{fig:pattern-1})--(\subref{fig:pattern-4}); see
  Fig.~\ref{fig:patterns}.  We note that in the rerouting, we only
  shorten vertical segments of top leaders and right segments of right
  leaders; hence the solution remains~$x$- and~$y$-separated.
  Moreover, in each step we decrease both $|\mathcal L|_x$ and
  $|\mathcal L|_y$.  Hence, after finitely many steps all patterns
  between top and right leaders have been removed without creating new
  patterns with other types of leaders.

  After all patterns have been removed, an $xy$-monotone curve connecting
  the top-right corner of~$R$ to~$o$, separating the top labels from
  the right labels, can be found as in the proof of Lemma~\ref{lem:xysep}.
\end{pf}

\subsection{Algorithm for the Three-Sided Case}
\label{sec:algorithm-three-sided-case}

In the three-sided case, we assume that the ports of the given instance~$I$
are located on three sides of~$R$; without loss of generality, on the left, top
and right side of~$R$. Basically, we solve a three-sided
instance by splitting the instance into two two-sided $L$-shaped instances that
can be solved independently; see Fig.~\ref{fig:three-sided:sketch}.

\begin{figure}[tb]
  \centering
  \begin{subfigure}{.45\textwidth}
     \centering
   \includegraphics[page=1]{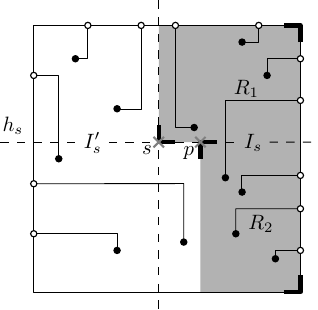}
   \caption{}\label{fig:three-sided:sketch}
 \end{subfigure}
   \begin{subfigure}{.45\textwidth}
     \centering
   \includegraphics[page=2]{three-sided2}
   \caption{}\label{fig:three-sided:first-direction}
 \end{subfigure}
\caption{\subref{fig:three-sided:sketch}) The three-sided instance partitioned
into two two-sided $L$-shaped instances~$I_{s}$ and~$I'_{s}$. The
instances are induced by the grid point~$s$ of $G$ and are balanced.
\subref{fig:three-sided:first-direction}) Illustration of the proof for
Lemma~\ref{lem:three-sided}. Assuming that the grid point $s$ of $G$, the
balanced instances~$I_{s}$ and~$I'_{s}$, and the curves~$C$ and~$C'$ are
given, a planar solution for the whole instance can be constructed.}
\label{fig:three-sided}
\end{figure}

Let~$G$ be the dual of the grid that
is induced by the sites and ports of the given
instance.  The idea is that each grid point~$s$ of~$G$
induces two two-sided $L$-shaped instances with some useful properties. We will
show that
there is a planar solution for~$I$ if and only if there is a grid point $s$ of
$G$ such that its induced two-sided instances both have planar solutions. Thus,
considering all~$O(n^2)$ grid points of~$G$ the problem reduces to solve those
$L$-shaped instances of the two-sided case. By means of a simple adaption of
the dynamic program presented in Section~\ref{sec:two-sided-algorithm} we solve
these instances in~$O(n^2)$ time achieving~$O(n^4)$ running time in total.

In the following we call horizontal and vertical lines through grid points of
$G$ \emph{horizontal} and \emph{vertical grid lines}, respectively.
We now define the two two-sided $L$-shaped instances~$I_s$ and~$I'_s$ of
a grid point~$s$ of~$G$ formally.  To that end, let~$R_1$ be
the rectangle that is spanned by the top-right corner of~$R$ and~$s$, and
let~$R_2(p)$ be the rectangle that is spanned by a point~$p$ on the horizontal
grid line~$h$ through~$s$ and the bottom-right corner of~$R$; see
Fig.~\ref{fig:three-sided:sketch}.
The instance~$I_s(p)$ contains all sites and ports in~$R_1\cup R_2(p)$
and~$I'_s(p)$ contains all sites and ports in~$R \setminus
(R_1\cup R_2(p))$. We say that~$I_s(p)$ and~$I'_s(p)$ are \emph{balanced}
if all right ports lie in~$R_1\cup R_2(p)$, all left ports lie in~$R\setminus
(R_1\cup R_2(p))$ and~$R_1\cup R_2(p)$ contains the same number of sites as it
contains ports. Since the number of ports and sites in~$I$ is equal, this
directly implies that~$R\setminus (R_1\cup R_2(p))$ contains the same number of
sites as it contains ports. In particular, the choice of balanced
instances~$I_s(p)$ and~$I'_s(p)$ for a grid point~$s$ of~$G$ is
unique with respect to the contained sites and ports; only the location
of~$p$ might differ. We can therefore write~$I_s$
and~$I'_s$ for balanced instances and $R_1$ and $R_2$ for their defining
rectangles.
For any solution of~$I_s$ and any solution of~$I'_s$,
we require that all leaders are completely contained in~$R_1\cup R_2$
and in~$R \setminus (R_1\cup R_2)$, respectively.  The next lemma states that a
three-sided instance $I$ has a planar solution if and only if it can be
partitioned into two two-sided $L$-shaped instances that have planar solutions.
To that end let $h_s$ denote the horizontal grid line through $s$.
Figure~\ref{fig:three-sided} illustrates the lemma.

\begin{lemma}\label{lem:three-sided}
  There is a planar solution~$\mathcal L$ for a three-sided instance~$I$ if and
  only if there is a grid point $s$ of~$G$ with balanced instances~$I_s$
  and~$I'_s$ over rectangles~$(R_1,R_2)$, an $xy$-monotone curve~$C$ from the
  top-right corner to the bottom-left corner of~$R$ and an
  $xy$-monotone curve~$C'$ from the top-left corner to the bottom-right corner
  of~$R$ such that
  \begin{compactenum}
   \item each point on~$C$ satisfies the strip condition with respect to the
      ports and sites in~$I_s$,
   \item $C$ contains the top-left corner of~$R_2$ and the intersection of $h_s$
with the left segment of~$R$,
  \item each point on~$C'$ satisfies the strip condition with respect to the
      ports and sites in~$I'_s$,
   \item $C'$ contains the top-left corner of~$R_2$ and the intersection
      of $h_s$ with the right segment of~$R$.
  \end{compactenum}
\end{lemma}

\begin{proof}
  First, assume that $s$, $I_{s}$, $I'_{s}$, $(R_1,R_2)$, $C$ and~$C'$
  exist as required; see Fig.~\ref{fig:three-sided:first-direction}. The curve~$C$
  partitions~$R$ into two regions; we denote the region above~$C$
  by~$A_1$ and the region below~$C$ by~$A_2$. By Lemma~\ref{lem:condit}, there
  is a planar solution~$\mathcal L_1$ for the sites and ports in~$A_1$ such
  that all leaders
  of~$\mathcal L_1$ lie in~$A_1$. Since~$C$ contains
  the top-left corner of~$R_2$ and does not cross~$h_s$ until it reaches the 
  intersection point of~$h_s$ with the left
  segment of~$R$, we know that all leaders of
  $\mathcal L_1$ are contained in~$R_1\cup R_2$. Analogously, there is a
  planar
  solution~$\mathcal L_2$ for the sites and ports in~$A_2$ such that all leaders
  of~$\mathcal L_2$ lie in~$A_2$. Consequently, we can combine~$\mathcal L_1$
  and~$\mathcal L_2$ into a planar solution~$\mathcal L_s$ for the sites and
  ports in~$I_s$. Using symmetric arguments, we obtain a planar
solution~$\mathcal
  L'_s$ for~$I'_s$. As~$I_s$ and~$I'_s$ are defined over complementary areas,
  the solutions~$\mathcal L_s$ and~$\mathcal L'_s$ can be combined into a planar
  solution of~$I$.

  \begin{figure}[tb]
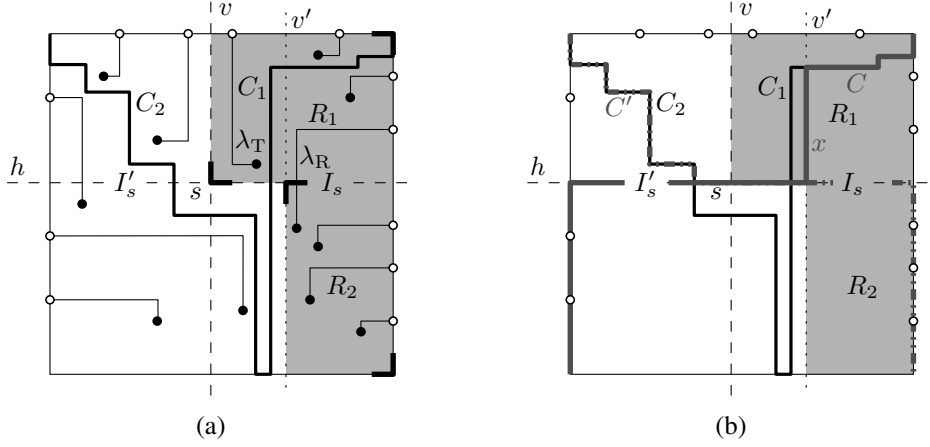

    \centering
    \begin{subfigure}{.45\textwidth}
       \centering
     \includegraphics[page=3]{three-sided2}
     \caption{}\label{fig:three-sided:partitioned}
   \end{subfigure}
     \begin{subfigure}{.45\textwidth}
       \centering
     \includegraphics[page=4]{three-sided2}
     \caption{}\label{fig:three-sided:curves}
   \end{subfigure}
   \caption{Illustration of the proof for
  Lemma~\ref{lem:three-sided}. It is assumed that the partitioned planar
  solution~$\mathcal L$ for the three-sided instance is given.
  \subref{fig:three-sided:partitioned}) By Proposition~\ref{prop:partitioned} we
  can assume that~$\mathcal L$ is partitioned by the curves~$C_1$ and~$C_2$. The
  extremal top leader~$\lambda_\TT$ induces the site~$s$ and the extremal
  right leader~$\lambda_\RR$ induces the line~$v'$.
  \subref{fig:three-sided:curves}) Based on $C_1$, $C_2$, $h$ and~$v'$, the
  curves~$C$ and~$C'$ can be constructed such that they do not cross any leader
  of~$\mathcal L$. }\label{fig:three-sided2}
  \end{figure}

  Assume that there is a planar solution~$\mathcal L$ for a three-sided
  instance~$I$; see Fig.~\ref{fig:three-sided2}. First, note that we can
imagine an instance of \textsc{Three-Sided
  Boundary Labeling} as a degenerated instance of \textsc{Four-Sided Boundary
  Labeling} with no bottom ports. Thus, Proposition~\ref{prop:partitioned} also holds for the
  three-sided case, when assuming that the four $xy$-monotone curves partitioning
  the solution meet on the bottom segment of~$R$. Hence, without
  loss of generality, we assume that~$\mathcal L$ is also partitioned by four
  $xy$-monotone curves $C_1$, $C_2$, $C_3$ and~$C_4$. In particular, let~$C_1$
  denote the curve
  that starts at the top-right corner of~$R$ and let~$C_2$ denote the curve that
  starts at the top-left corner of~$R$; see
  Fig.~\ref{fig:three-sided:partitioned}. The curves~$C_3$ and~$C_4$ are
  completely contained in the bottom side of~$R$ and can therefore be
  omitted. We first show how to construct the grid point~$s$ and the
  instances~$I_{s}$ and~$I'_{s}$ such that they are balanced. Afterwards, we
  explain how to obtain~$C$ and~$C'$ from~$C_1$ and~$C_2$, respectively. Finally,
  we prove that each point on $C$ and $C'$ satisfies the strip condition with
  respect to $I_s$ and $I'_s$, respectively.

  Let~$\lambda_\TT$ be the top leader in~$\mathcal L$ with the longest vertical segment
  of all top leaders in~$\mathcal L$. In case the site of~$\lambda_\TT$
  lies to the right of~$\bend(\lambda_\TT)$, let~$v$ be the rightmost vertical
  grid line that lies to the left
  of~$\lambda_\TT$, and otherwise if the site of~$\lambda_\TT$ lies to the
  left of~$\bend(\lambda_\TT)$, let $v$ be the leftmost vertical grid line that
  lies to the right of~$\lambda_\TT$.  Furthermore, let~$h$
  be the topmost horizontal grid line that lies
below~$\bend(\lambda_\TT)$; see
  Fig.~\ref{fig:three-sided:partitioned}. Due to the choice of~$h$ and~$v$ all
  top leaders lie above~$h$ and none of them intersects~$h$ or~$v$. Furthermore, no
  right  or left leader of~$\mathcal L$ intersects $v$ above~$h$.
   The desired grid point $s$ is then the
  intersection point of $h$ and $v$.

  Now, let~$\lambda_\RR$ be the right leader in~$\mathcal L$ with longest
  horizontal segment among all right leaders in~$\mathcal L$ and let~$v'$ be
  the rightmost vertical grid line that lies to the left
  of~$\bend(\lambda_\RR)$. Note that $v'$ cannot be intersected by a left
  or a right leader, because both leader types are~$x$-separated. We
  define~$R_1$ to be the rectangle that is spanned by the top-right corner
  of~$R$ and $s$. Also, we define~$R_2$ to be the rectangle spanned by the
  bottom-right corner of~$R$ and the intersection point of $v'$ and $h$.
  The instance~$I_{s}$ is defined by~$R_1 \cup R_2$ and the instance~$I'_{s}$
  by~\ADD{$R \setminus (R_1\cup R_2)$}. We show that~$I_s$ and~$I'_s$ are
  balanced. To that end, we prove that a leader of~$\mathcal L$ is
  either completely contained in~$R_1\cup R_2$ or in~\ADD{$R \setminus (R_1 \cup R_2)$},
  that~$R_1\cup R_2$ contains only right and top leaders, and that~\ADD{$R \setminus (R_1
  \cup R_2)$} contains only left and top leaders.

  Due to the choice of~$v'$, all right leaders lie to the right of $v'$.
  Moreover, all right leaders whose site or port lies above~$h$, must lie to
  the right of~$v$, because by
  definition of~$v$ no right leader intersects $v$ above~$h$ \ADD{(otherwise it would intersect~$\lambda_T$)}, and because
  otherwise $C_1$ could not be an
  $xy$-monotone curve separating right and top  leaders.
  Thus, all right leaders lie in~$R_1\cup R_2$.
  For left leaders we can argue similarly. Since left and right leaders
  of~$\mathcal L$ are $x$-separated, all left leaders lie to the left of $v'$.
  All left leaders whose site or port lies above~$h$, must lie to the
  left of $v$, because by definition of~$v$ no left leader intersects~$v$
  above~$h$, and because otherwise~$C_2$ could  not be an~$xy$-monotone curve
  separating left from top leaders. Thus, all left leaders lie in~\ADD{$R\setminus
  (R_1\cup R_2)$}.
  Finally, consider the top leaders in~$\mathcal L$. By definition of~$h$
  and~$v$,  none of the top leaders intersects $h$ or $v$. In particular all
  top leaders lie above $h$ and cannot intersect~$R_2$. Consequently, each
  top leader is either contained in $R_1$ or in $R \setminus (R_1\cup R_2)$.
  This concludes the argument that~$I_{s}$ and~$I'_{s}$ are balanced. 

  We are left with the construction of the curves~$C$ and $C'$; see Fig.~\ref{fig:three-sided:curves}. 
  The curve~$C$ is derived from~$C_1$ as follows. Starting at the top-right corner of~$R$,
  the curve~$C$ coincides with~$C_1$ until~$C_1$ intersects~$h$ or~$v'$ above~$h$.
  If $C$ intersects~$v'$ above~$h$, it follows~$v'$ downwards until it
  hits~$h$. Then, in both cases, it follows~$h$ until~$h$ intersects the left
  segment of~$R$. Finally,~$C$ follows the left segment of~$R$ to the
  bottom-left corner of~$R$. The curve~$C'$ is constructed symmetrically.

  By construction,~$C$ contains the top-left corner of~$R_2$ and the intersection
  point of~$h$ with the left segment of~$R$. Symmetrically,~$C'$ contains the
  top-left corner of~$R_2$ and the intersection point of~$h$ with the right
  segment of~$R$. We finally show that each point on~$C$ satisfies the strip
  condition with respect to the sites and ports in~$I_{s}$. Using symmetric
  arguments we can prove the analogous statement for~$C'$ and~$I'_s$.

  By the previous reasoning, we know that each leader of~$\mathcal L$
  either lies completely inside or completely outside of~$R_1\cup
  R_2$.  Each leader that lies in~$R_1\cup R_2$ is either a
  top or a right leader and does not intersect~$C$. Otherwise, if such
  a leader intersected~$C$, it would also intersect~$C_1$ or the
  segment~$x$ of~$v'$ that is contained in~$C$. In particular,~$x$
  cannot be intersected by any leader because it lies to the left of
  all right leaders and below~$C_1$. Thus, the leaders in~$R_1\cup
  R_2$ form a planar solution for~$I_{s}$ without
  intersecting~$C$. Hence, the claim directly follows from
  Lemma~\ref{lem:condit}.
\end{proof}

\begin{figure}[tb]
  \centering
   \includegraphics[page=5]{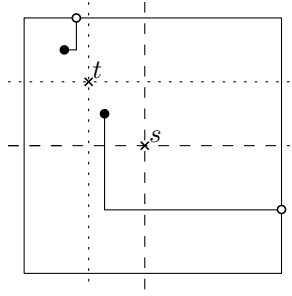}
   \caption{There are no balanced instances~$I_{s}$ and~$I'_{s}$ for the
grid point~$s$. However, by Lemma~\ref{lem:three-sided} there must be
another grid point~$t$ with balanced instances~$I_{t}$
and~$I'_{t}$ if the
instance has a planar solution. }
   \label{fig:three-sided:no-balanced-instances}
\end{figure}

Our approach now works as follows. For each grid point~$s$ of~$G$ we
compute the instances~$I_{s}$ and~$I'_{s}$ such that they are
balanced.  Then, by Lemma~\ref{lem:three-sided}, we can apply our
algorithm presented in Section~\ref{sec:two-sided-algorithm} in order
to solve~$I_{s}$ and~$I'_{s}$ independently.  To that end, we slightly
adapt the dynamic program such that it only considers curves
satisfying the properties required by Lemma~\ref{lem:three-sided}.  If
both instances can be solved, we combine these solutions into one
solution and return that solution as the final result. Otherwise, we
continue with the next grid point of~$G$. If all grid points of~$G$
have been explored without finding a planar solution, the algorithm
decides that there is no planar solution.

Note that it may happen that, for a grid point $s$, there are no balanced
instances~$I_{s}$ and~$I'_{s}$; for an example see
Fig.~\ref{fig:three-sided:no-balanced-instances}. However, in that case, if~$I$
has a solution, we also know by Lemma~\ref{lem:three-sided} that there is
another grid point~$t$ such that for~$t$ we find balanced instances.
Hence, we can refrain from considering~$s$.

Creating the two instances~$I_{s}$ and~$I'_{s}$ for a grid point~$s$ 
takes linear time, if we assume that the sites are sorted by their
$x$-coordinates. By Theorem~\ref{thm:efficient-twosided-labeling} we
then need~$O(n^2)$ time and~$O(n)$ space to solve~$I_{s}$ and~$I'_{s}$.
Consequently, we need~$O(n^2)$
time and~$O(n)$ space to process a single grid point~$s$. Since we
consider~$O(n^2)$ grid points, the following theorem follows.

\begin{theorem}\label{thm:efficient-threesided-labeling}
  \textsc{Three-Sided Boundary Labeling} can be solved in~$O(n^{4})$
  time using~$O(n)$ space.
\end{theorem}

Note that, except for the length minimization, our approach for the three-sided
case also carries over to the extensions from
Section~\ref{sec:two-sided-extensions}, because we only solve subinstances
of~\textsc{Two-Sided Boundary Labeling with Adjacent Sides}. In particular with
corresponding impact on the running time we can soften the restriction that the number
of labels and sites is equal.

\subsection{Algorithm for the Four-Sided Case}
\label{sec:algorithm-four-sided-case}

In this section, we consider the case that the ports lie on all four
sides of~$R$.  The main idea is to seek a partitioned solution, which
exists by Proposition~\ref{prop:partitioned}. For a given partitioned
solution~$\mathcal L$, we call a leader \emph{extremal} if all other
leaders of the same type in~$\mathcal L$ have shorter orthogonal
segments; recall that the \emph{orthogonal} segment of a $po$-leader
is the segment connecting the bend to the port.  The algorithm
consists of two steps. First, we explore all choices of
(non-overlapping) extremal leaders~$\lambda_\LL$ and~$\lambda_\RR$ for
the left and the right side of~$R$, respectively, plus a horizontal
line~$h$ that separates the top leaders and the bottom leaders. This
information splits the instance into two independent three-sided
instances; see Fig.~\ref{fig:separation-four-sided}.  There are,
however, two crucial differences from a usual three-sided
instance. First, one side of the instance is not a straight-line
segment but an $x$-monotone orthogonal curve~$C$ that is defined
by~$\lambda_\LL, \lambda_\RR$ and~$h$.  Second, the extremal positions
of~$\lambda_\LL$ and~$\lambda_\RR$ imply a separation of the points
that are labeled from the left and the right side. Let~$I^3_1$ be the
three-sided instance above~$C$ and let~$I^3_2$ be the three-sided
instance below~$C$. The algorithm solves~$I^3_1$ and~$I^3_2$
independently from each other. If for at least one of the two
instances there is no solution, the algorithm continues with the next
choice of~$\lambda_\LL$, $\lambda_\RR$ and~$h$. Otherwise, it combines
the planar solutions of~$I^3_1$ and~$I^3_2$ into one planar solution and
returns this solution. In case that all choices of~$\lambda_\LL$,
$\lambda_\RR$ and~$h$ have been explored without finding a solution,
the algorithm returns that there is no planar solution.

We next describe how to solve the three-sided instance~$I^3_1$.  A
symmetric approach can be applied to~$I^3_2$. The algorithm explores
all choices of the extremal leader~$\lambda_{\TT}$ for the top side
of~$R$. This extremal leader partitions the instance into two
two-sided subinstances~$I^2_1$ and~$I^2_2$ as follows.
Let~$A_{\TT,\RR}$ be the rectangle that is spanned
by~$\bend(\lambda_\TT)$ and the top-right corner of~$R$; see
Fig.~\ref{fig:separation-four-sided:areas}. Analogously,
let~$A_{\TT,\LL}$ be the rectangle that is spanned
by~$\bend(\lambda_\TT)$ and the top-left corner of~$R$.  Analogously,
for~$\lambda_\RR$ we define the area~$A_{\RR,\TT}$ to be the rectangle
that is spanned by~$\bend(\lambda_\RR)$ and the top-right corner
of~$R$, and~$A_{\RR,\BB}$ to be the rectangle spanned
by~$\bend(\lambda_\RR)$ and the bottom-right corner of~$R$. For the
leader~$\lambda_\LL$ we define~$A_{\LL,\BB}$ to be the rectangle
spanned by~$\bend(\lambda_\LL)$ and the bottom-left corner of~$R$,
and~$A_{\LL,\TT}$ to be the rectangle spanned by~$\bend(\lambda_\LL)$
and the top-left corner of~$R$. We assume that the port~$p$
of~$\lambda_\TT$ is only contained in that area~$A \in
\{A_{\TT,\RR},A_{\TT,\LL}\}$ that also contains the site
of~$\lambda_\TT$. We make analogous assumptions for~$\lambda_\LL$
and~$\lambda_\RR$.

The instance~$I^2_1$ consists of all ports and sites
in~$A_1=(A_{\RR,\TT}\cup A_{\TT,\RR}) \setminus (A_{\TT,\LL}\cup
A_{\RR,\BB})$, and~$I^2_2$ consists of all ports and sites
in~$A_2=(A_{\LL,\TT}\cup A_{\TT,\LL}) \setminus (A_{\TT,\RR}\cup
A_{\LL,\BB})$; see Fig.~\ref{fig:four-sided-instances}.  We
solve~$I^2_1$ and~$I^2_2$ independently from each other using the
dynamic program introduced in Section~\ref{sec:two-sided-algorithm}
for each instance. However, we enforce that it only considers
$xy$-monotone curves that exclude top leaders crossing the horizontal
line through~$\bend(\lambda_\TT)$, left leaders crossing the vertical
line through~$\bend(\lambda_\LL)$ and right leaders crossing the
vertical line through~$\bend(\lambda_\RR)$.
If for at least one of the two
instances there is no solution, the algorithm continues to
explore the next choice of~$\lambda_{\TT}$. Otherwise, it combines the
solutions of~$I^2_1$ and~$I^2_2$ into one solution and returns the
result as the solution of~$I^3_1$. In case that all choices of
$\lambda_{\TT}$ have been explored without finding a solution, the
algorithm returns that there is no solution for the given three-sided
instance. The following lemma shows that the algorithm is correct.

\begin{figure}[tb]
  \centering
  \begin{subfigure}{.30\textwidth}
    \centering
    \includegraphics[scale=1.0,page=1]{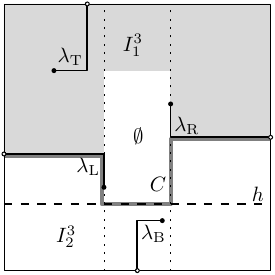}
    \caption{}\label{fig:separation-four-sided}
  \end{subfigure}
  \hspace{3mm}
  \begin{subfigure}{.30\textwidth}
    \centering
    \includegraphics[scale=1.0,page=3]{separation-four-sided}
    \caption{}\label{fig:separation-four-sided:areas}
  \end{subfigure}
  \hspace{3mm}
  \begin{subfigure}{.30\textwidth}
    \centering
    \includegraphics[scale=1.0,page=2]{separation-four-sided}
    \caption{}\label{fig:four-sided-instances}
  \end{subfigure}
  \caption{\subref{fig:separation-four-sided}) The right leader
    $\lambda_\LL$, the left leader $\lambda_\RR$ and the horizontal
    line $h$ split the instance into two three-sided instances $I^3_1$
    and $I^3_2$.  \subref{fig:separation-four-sided}) Sketch of the
    areas $A_{\TT,\LL}$, $A_{\TT,\RR}$, $A_{\RR,\TT}$, $A_{\RR,\BB}$,
    $A_{\LL,\TT}$ and $A_{\LL,\BB}$.
    \subref{fig:four-sided-instances}) The leaders $\lambda_\LL$,
    $\lambda_\RR$ and $\lambda_\TT$ split the three-sided instance
    into two two-sided instances.}
\end{figure}

\begin{lemma}
  Given an instance~$I$ of \textsc{Four-Sided Boundary Labeling},
  the following two statements are true.
  \begin{compactenum}
    \item If there is no planar solution for~$I$, the algorithm
      states this.
    \item If there is a planar solution for~$I$, the algorithm
      returns such a solution.
  \end{compactenum}
\end{lemma}

\begin{pf}
  In case the algorithm returns a solution, it has been constructed from
  planar solutions of disjoint instances of \textsc{Two-Sided Boundary
    Labeling with Adjacent Sides}. As the union of these two-sided
  instances contains all sites and ports of~$I$, the algorithm returns
  a planar solution of~$I$, which shows the first statement.

  Conversely, assume that~$I$ has a planar solution~$\mathcal L$.  By
  Proposition~\ref{prop:partitioned}, we may assume that~$\mathcal L$
  is partitioned.  In particular, let~$\lambda_{\TT}$, $\lambda_\LL$,
  $\lambda_\BB$ and~$\lambda_\RR$ be the extremal leaders in~$\mathcal
  L$ of the top, left, bottom and right side of~$R$, respectively, and
  let~$h$ be a horizontal line that separates the top leaders from the
  bottom leaders.

  Obviously,~$\lambda_\LL$,~$\lambda_\RR$ and~$h$ split the instance into
  two three-sided instances~$I^3_1$ and~$I^3_2$. As the algorithm
  systematically explores all choices of extremal right leaders,
  extremal left leaders and horizontal lines partitioning the set of
  sites, it must find~$\lambda_\LL$,~$\lambda_\RR$ and a horizontal
  line~$h'$ that separates the same sets of sites as~$h$.
  Thus,~$I^3_1$ and~$I^3_2$ are considered by the algorithm.

  Let~$I^3_1$ be the instance above the curve defined by~$\lambda_\LL$,
  $\lambda_\RR$ and~$h'$, and let~$I^3_2$ be the instance below that
  curve. We now show that the algorithm finds a planar solution
  for~$I^3_1$. Symmetric arguments hold for~$I^3_2$.
  As the algorithm explores all choices of extremal top leaders
  in~$I^3_1$, it also considers~$\lambda_{\TT}$ to be the
  extremal top leader. This leader partitions the area
  of~$I^3_1$ into the two disjoint areas~$A_1=(A_{\RR,\TT}\cup A_{\TT,\RR})
  \setminus (A_{\TT,\LL}\cup A_{\RR,\BB})$ and~$A_2=(A_{\LL,\TT}\cup
A_{\TT,\LL})
  \setminus (A_{\TT,\RR}\cup A_{\LL,\BB})$; see
  Fig.~\ref{fig:separation-four-sided}. It directly follows from the
  extremal choice of~$\lambda_\RR$, $\lambda_{\TT}$ and~$\lambda_\LL$ that
  there is no leader in~$\mathcal L$ that intersects both~$A_1$ to~$A_2$. 
  In particular, no left leader intersects~$A_1$ and no
  right leader intersects~$A_2$. Thus,~$A_1$ and~$A_2$ split~$\mathcal L$
  into independent planar solutions~$\mathcal L_1$ and~$\mathcal L_2$
  of two two-sided instances~$I^2_1$ and~$I^2_2$ induced by~$A_1$
  and~$A_2$, respectively. Note that the algorithm considers the same
  two-sided instances independently from each other. As~$I^2_1$ has a
  solution, namely~$\mathcal L_1$, we know that the dynamic program
  finds a solution~$\mathcal L^2_1$ for~$I^2_1$. In particular, all
  leaders of ~$\mathcal L^2_1$ lie in $A_1$.

  Applying symmetric arguments for~$I^2_2$, the algorithm yields a planar
  solution~$\mathcal L^2_2$ that stays in~$A_2$. Consequently,
  combining~$\mathcal L^2_1$ and~$\mathcal L^2_2$ into one solution
  yields a planar solution~$\mathcal L^3_1$ for~$I^3_1$. Analogously,
  we obtain a planar solution~$\mathcal L^3_2$ for~$I^3_2$. Obviously,
  due to the separation by~$\lambda_\LL$,~$\lambda_\RR$ and~$h'$, the union
  of~$\mathcal L^3_1$ and~$\mathcal L^3_2$ is also planar, which is the
  overall solution returned by the algorithm. This proves the second
  statement of the lemma.
\end{pf}

Let us analyze the running time of the algorithm. Obviously, there are~$O(n^5)$
possible combinations of left and right extremal leaders and a
horizontal line separating the top and bottom-labeled sites.  For each
combination, we independently solve two three-sided instances.  For
such a three-sided instance, we consider $O(n^2)$ choices for the
extremal leader $\lambda_\TT$ and independently solve two independent
two-sided instances with Theorem~\ref{thm:efficient-twosided-labeling}
in $O(n^2)$ time.  This implies that solving one three-sided instances
takes $O(n^4)$ time.  Thus, the overall running time is~$O(n^9)$.  The
following theorem summarizes this result.

\begin{theorem}\label{thm:efficient-foursided-labeling}
  \textsc{Four-Sided Boundary Labeling} can be solved in~$O(n^{9})$
  time using~$O(n)$ space.
\end{theorem}

Note that, except for the length minimization, our approach for the four-sided
case also carries over to the extensions from
Section~\ref{sec:two-sided-extensions}, because we only solve subinstances
of~\textsc{Two-Sided Boundary Labeling with Adjacent Sides}. In particular with
corresponding impact on the running time we can soften the restriction that the number
of labels and sites is equal.

\section{Conclusion}
\label{sec:conclusion}

In this paper, we have studied the problem of testing whether an
instance of {\sc Two-Sided Boundary Labeling with Adjacent Sides}
admits a planar solution. We have given the first efficient algorithm
for this problem, running in~$O(n^2)$ time.

Our algorithm can also be used to solve a variety of different
extensions of the problem. We have shown
how to generalize to sliding ports instead of fixed ports without
increasing the running time and how to maximize the number of labeled
sites such that the solution is planar
in~$O(n^3\log n)$ time. We further have given an extension to the
algorithm that minimizes the total leader length in~$O(n^8\log n)$ time.

With some additional work, our approach can also be
used to solve {\sc Three-Sided} and {\sc Four-Sided Boundary Labeling}
in polynomial time. We have introduced an algorithm solving the
three-sided case in~$O(n^4)$ time and the four-sided case
in~$O(n^{9})$ time. Also, except for the length minimization, all
extensions carry over.  It remains open whether a minimum
length solution of  {\sc Three-Sided} and {\sc Four-Sided Boundary
Labeling} can be computed in polynomial time.

%\clearpage
\addcontentsline{toc}{section}{\numberline{}\refname}
\bibliographystyle{abbrv} % {abbrv} % {plain}
\bibliography{abbrv,maplab,boundary}

\end{document}